\documentclass[preprint,3p]{elsarticle}

\usepackage[english]{babel}
\usepackage{amsfonts}
\usepackage{amsmath}
\usepackage{amssymb}                                                                                                                
\usepackage{amsthm}
\usepackage[utf8]{inputenc}                                    
\usepackage[T1]{fontenc}
\usepackage{mathrsfs}
\usepackage[small,bf]{caption}
\usepackage{subfig}
\usepackage{stmaryrd}
\usepackage{float}
\usepackage{changepage}
\usepackage{enumitem}

\journal{arXiv}

\newcommand{\eps}{\varepsilon}
\DeclareMathOperator*{\supp}{supp}

\makeatletter
\DeclareFontFamily{OMX}{MnSymbolE}{}
\DeclareSymbolFont{MnLargeSymbols}{OMX}{MnSymbolE}{m}{n}
\SetSymbolFont{MnLargeSymbols}{bold}{OMX}{MnSymbolE}{b}{n}
\DeclareFontShape{OMX}{MnSymbolE}{m}{n}{
    <-6>  MnSymbolE5
   <6-7>  MnSymbolE6
   <7-8>  MnSymbolE7
   <8-9>  MnSymbolE8
   <9-10> MnSymbolE9
  <10-12> MnSymbolE10
  <12->   MnSymbolE12
}{}
\DeclareFontShape{OMX}{MnSymbolE}{b}{n}{
    <-6>  MnSymbolE-Bold5
   <6-7>  MnSymbolE-Bold6
   <7-8>  MnSymbolE-Bold7
   <8-9>  MnSymbolE-Bold8
   <9-10> MnSymbolE-Bold9
  <10-12> MnSymbolE-Bold10
  <12->   MnSymbolE-Bold12
}{}

\let\llangle\@undefined
\let\rrangle\@undefined
\DeclareMathDelimiter{\llangle}{\mathopen}%
                     {MnLargeSymbols}{'164}{MnLargeSymbols}{'164}
\DeclareMathDelimiter{\rrangle}{\mathclose}%
                     {MnLargeSymbols}{'171}{MnLargeSymbols}{'171}
\makeatother

\newtheorem{theorem}{Theorem}[section]
\newtheorem{lemma}[theorem]{Lemma}

\newtheorem{corollary}[theorem]{Corollary}
\newdefinition{definition}[theorem]{Definition}
\newdefinition{notation}[theorem]{Notation}
\newdefinition{remark}[theorem]{Remark}
\newdefinition{example}[theorem]{Example}

\begin{document}
\begin{frontmatter}

\title{Bideterministic Weighted Automata\tnoteref{extversion}}
\tnotetext[extversion]{This is an extended version of the article \cite{kostolanyi2022b} published in the proceedings of the conference CAI 2022.}

\author{Peter Kostol\'anyi\fnref{funding}}
\ead{kostolanyi@fmph.uniba.sk}
\address{Department of Computer Science, Comenius University in Bratislava, \\
Mlynsk\'a dolina, 842 48 Bratislava, Slovakia}
\fntext[funding]{The author was supported by the grant VEGA 1/0601/20.}

\begin{abstract}
A finite automaton is called bideterministic if it is both deterministic and codeterministic -- that is, if it is deterministic and~its transpose is deterministic as well. 
The study of such automata in~a~weighted setting is~initiated. All trim bideterministic weighted automata over
integral domains and~over positive semirings are proved to be minimal. On~the~contrary, it is observed that this property does not hold over
commutative rings in general: non-minimal trim bideterministic weighted automata do exist over all semirings that are not zero-divisor free,
and~over many~such semirings, these automata might not even admit equivalents that are both minimal and~bideterministic. The~problem of~determining whether a given rational series
is realised by~a~bideterministic automaton is shown to be decidable over fields and~over tropical semi\-rings. An~example of~a~positive semiring over which this problem
becomes undecidable is given as well.
\end{abstract}

\end{frontmatter}

\section{Introduction}
Unlike the classical nondeterministic finite automata without weights, weighted finite automata might not always be
determinisable. Nevertheless, partly due to their relevance for~applications such as natural language and~speech processing~\cite{mohri1997a}
and~partly due~to~their purely theoretical importance, deterministic weighted automata and~the~questions related to them -- such as decidability of~determinisability,
existence of efficient determinisation algorithms, or~characterisations of~rational series realised by deterministic weighted automata -- have received significant attention.
Deterministic weighted automata were studied over specific classes of~semirings, such as tropical semirings or~fields~\cite{allauzen2003a,bell2021a,bell2023a,kirsten2009a,kirsten2005a,kostolanyi2022a,lombardy2006a,mohri1997a,mohri2009a}, as well as over strong bimonoids~\cite{ciric2010a}, often under certain additional restrictions.

Determinism in~weighted automata has been a~rich source of~intriguing questions, many of~which remain unsettled. For instance, the~decidability status of~the~determinisability problem for weighted
automata is, despite some partial results \cite{kirsten2009a,kirsten2005a,lombardy2006a}, still open over tropical semi\-rings~\cite{lombardy2021a,lombardy2006a}. 
Decidability of~the~same problem for~automata over fields -- and~in~particular, over the~rational numbers -- remained unknown for~a~long time as well~\cite{lombardy2006a,kostolanyi2022a}, and~the~problem was only recently proved decidable by~J.~P.~Bell and~D.~Smertnig~\cite{bell2023a}. 
Complexity of~this problem stays open~\cite{bell2023a}.
It~thus makes sense to take a~look at~stronger~forms of determinism in~weighted automata, which may be amenable~to a~somewhat easier analysis.\goodbreak

Within these lines, deterministic weighted automata with certain additional requirements
on their weights have been studied. This includes, for~instance, the~research on~crisp-deterministic weighted automata by~M.~\'Ciri\'c~et~al.~\cite{ciric2010a}. Another less explored possibility is to examine the~weighted counterpart of~some 
particularly simple subclass of~deterministic finite automata without weights -- that is, to impose further restrictions not only on~weights of~deterministic weighted automata,
but on~the~concept of~determinism itself. This is a~direction that we~follow in~this article.\goodbreak            

More tangibly, this article undertakes the~study of~\emph{bideterministic} finite automata in~the~weighted setting.
A~finite automaton is bideterministic if it is both deterministic and codeterministic -- the~latter meaning that the~transpose of~the~automaton, obtained by reversing all transitions and exchanging the~roles of~initial and~terminal states,
is~deterministic as~well. This in~particular implies that a~bideterministic automaton always contains at~most one initial and~at~most one terminal state.
Bideterministic finite automata have been first touched upon from a~theoretical perspective by~\mbox{J.-\'E.}~Pin~\cite{pin1992a}, as~a~particular case of~reversible finite automata. 
The fundamental properties of bideterministic finite automata have later been explored mostly by~H.~Tamm and~E.~Ukkonen~\cite{tamm2003a,tamm2004a}, who have shown
that a~trim bideterministic automaton is~always a~minimal nondeterministic automaton for the language it recognises -- in fact, it is~the~\emph{only} minimal nondeterministic finite automaton recognising its language. 
Minimality is understood here in~the~strong sense, \emph{i.e.}, with respect
to the number of states. An~alternative proof of~this~minimality property of~trim bideterministic automata was recently presented by~R.~S.~R.~Myers, S.~Milius, and~H.~Urbat~\cite{myers2021a}.
Transition minimality of~bideterministic finite automata has also been established~\cite{tamm2008a}.
 
In addition to the~above-mentioned studies, bideterministic automata have been -- explicitly or implicitly~-- considered in~connection to the~star height problem~\cite{mcnaughton1967a,mcnaughton1969a,gruber2012a},
from the perspective of~language inference~\cite{angluin1982a}, in~the~theory of~block codes~\cite{shankar2003a}, and~within the~study of~presentations of~inverse monoids~\cite{stephen1990a,janin2015a}.

We define \emph{bideterministic weighted automata} over a~semiring by analogy to their unweighted counterparts,
and study the~conditions under which the~fundamental property of H.~Tamm and~E.~Ukkonen~\cite{tamm2003a,tamm2004a} generalises to~the~weighted setting.
Thus, given a~semiring~$S$, we ask the following questions: Are all trim bideterministic weighted automata over $S$ minimal? Does every bideterministic
automaton over~$S$ admit a~bideterministic equivalent that is at~the~same time minimal? We answer both these questions in~affirmative when~$S$
is an integral domain or a positive -- \emph{i.e.}, both zero-sum free and zero-divisor free -- semi\-ring. On~the~other hand, we show that 
the~answer to the~former question is negative whenever~$S$ is not zero-divisor free -- and~thus also whenever it is a~nontrivial commutative ring
other than an~integral domain. Moreover, we show that the answer to~the~latter question is negative for a~large class of~commutative semirings including a~multitude of~\emph{finite}
commutative \emph{rings}, and~fully understand the~answers to this question over the~rings $\mathbb{Z}/m\mathbb{Z}$ with $m \in \mathbb{N}$. 

We also consider the \emph{bideterminisability problem}, in which the task is to decide whether a~given weighted automaton over a~semiring~$S$ admits
a~bideterministic equivalent over~$S$. We show that the~problem is decidable in polynomial time when~$S$ is a~field, as it is sufficient to simply minimise the~input automaton
and~check whether the~resulting automaton is bideterministic; on~the~other hand, we give examples showing that this simple procedure might not work over integral domains in~general.
Next, we establish decidability of~the~bideterminisability problem over tropical semirings (of~nonnegative integers, integers, and~rationals).
These results suggest that the~bideterminisability problem for~weighted automata might be somewhat easier than the~determinisability problem, whose decidability status over tropical semirings remains open~\cite{lombardy2021a,lombardy2006a},
and~for~which only algorithms lacking efficiency~\cite{bell2023a} or~generality~\cite{kostolanyi2022a} are known in~the~case of~fields.
Finally, we give an~example of~a~positive semiring over which the~bideterminisability problem is undecidable.

\section{Preliminaries}
\label{sec:prelim}

We denote by $\mathbb{N}$, $\mathbb{Z}$, and $\mathbb{Q}$, respectively, the sets of all \emph{nonnegative} integers, integers, and~rational numbers.

A \emph{semiring} is a quintuple $(S,+,\cdot,0,1)$ such that $(S,+,0)$ is a commutative monoid, $(S,\cdot,1)$ is a monoid, multiplication distributes over addition both from left and from right,
and $a \cdot 0 = 0 \cdot a = 0$ holds for all $a \in S$; it is said to~be \emph{commutative} when~$\cdot$ is.
A~semiring~$S$ is~\emph{zero-sum~free}~\cite{golan1999a,hebisch1998a} if $a + b = 0$ for some $a,b \in S$
implies $a = b = 0$ and \emph{zero-divisor~free}~\cite{hebisch1998a}, or \emph{entire} \cite{golan1999a}, if $a\cdot b = 0$ for~some $a,b \in S$ implies that $a = 0$ or $b = 0$.
A semiring is \emph{positive} \cite{eilenberg1974a,kirsten2014a} if it is both zero-sum free and zero-divisor free. 
A~\emph{ring} is a~semiring $(R,+,\cdot,0,1)$ such that~$R$ forms an~abelian group with addition.
An~\emph{integral domain} is a~nontrivial zero-divisor free commutative ring.
A~\emph{field} is an~integral domain $(\mathbb{F},+,\cdot,0,1)$ such that $\mathbb{F} \setminus \{0\}$ forms an~abelian group with multiplication.\goodbreak

The \emph{ring of all polynomials} over a~commutative ring~$R$ in indeterminates $x_1,\ldots,x_n$ for $n \in \mathbb{N} \setminus \{0\}$ is denoted by~$R[x_1,\ldots,x_n]$. For $\mathbb{F}$ a field, $\mathbb{F}(x_1,\ldots,x_n)$
denotes the field of all \emph{rational fractions} over $\mathbb{F}$ in~$x_1,\ldots,x_n$, \emph{i.e.}, the fraction field of $\mathbb{F}[x_1,\ldots,x_n]$.
The subring of~a~ring~$R$ generated by~a~set~$X \subseteq R$ is the~smallest subring of~$R$ containing~$X$. Given commutative rings $R$, $R'$ such that $R \subseteq R'$ and~$X \subseteq R'$,
we denote by~$R[X]$ the~\emph{subring of~$R'$ generated by~$X$ over~$R$}, \emph{i.e.}, the~subring of~$R'$ generated by~$R \cup X$.\goodbreak  

An \emph{ideal} of~a~semiring~$S$ is a~nonempty set $I \subseteq S$ such that $a + b$, $as$, and $sa$ are in~$I$ for all $a,b \in I$ and~$s \in S$.
Given $X \subseteq S$, we denote by $(X)$ the~\emph{ideal of $S$ generated by $X$}, \emph{i.e.}, the~smallest ideal of $S$ containing $X$.
If $X = \{x_1,\ldots,x_n\}$ is a~finite set, we write $(x_1,\ldots,x_n)$ instead of~$(\{x_1,\ldots,x_n\})$. 
The~\emph{quotient ring} of~a~\emph{ring}~$R$ by~an~ideal~$I$ of~$R$ is denoted by $R/I$.

Alphabets are assumed to be finite and~nonempty. Deterministic finite automata (without weights) are~understood to~have partial transition functions in~this article,
which is a~natural assumption in~the~context of~bideterministic automata.
More precisely, we identify deterministic finite automata with nondeterministic finite automata having at most one initial state such that there is at~most one transition upon each
letter leading from each state. In~particular, a~deterministic finite automaton defined like this can~be empty. Minimisation of such automata works essentially in~the~same way as for~automata 
with complete transition functions -- the~only difference is that there is no~dead state in~the~minimised automaton.     
                                                                   
We now briefly recall some basic facts about noncommutative formal power series and weighted automata. 
More information can be found in \cite{berstel2011a,droste2009a,droste2021a,sakarovitch2009a}.

A \emph{formal power series} over a semiring $S$ and alphabet $\Sigma$ is a mapping $r\colon \Sigma^* \to S$. The value of $r$ upon $w \in \Sigma^*$ is usually denoted by $(r, w)$ and called
the \emph{coefficient} of $r$ at $w$; the coefficient of $r$ at $\eps$, the empty word, is referred to as the \emph{constant coefficient}. The series $r$ itself is written as 
\begin{displaymath}
r = \sum_{w \in \Sigma^*} (r,w)\,w.
\end{displaymath}
The set of all formal power series over $S$ and $\Sigma$ is denoted by $S\llangle\Sigma^*\rrangle$.\goodbreak

Given series $r,s \in S\llangle\Sigma^*\rrangle$, their \emph{sum} $r + s$ and \emph{product} $r \cdot s$ are defined by
$(r + s, w) = (r, w) + (s, w)$ and
\begin{displaymath}
(r\cdot s,w) = \sum_{\substack{u,v \in \Sigma^* \\ uv = w}} (r, u)(s, v)
\end{displaymath}
for all $w \in \Sigma^*$. Every $a \in S$ is identified with a series with constant coefficient $a$ and all other coefficients zero,
and every $w \in \Sigma^*$ with a series with coefficient $1$ at $w$ and zero coefficients at all $x \in \Sigma^* \setminus \{w\}$.
Thus, for~instance, $r = 2ab + 3abb$ is a series with $(r, ab) = 2$, $(r, abb) = 3$, and~$(r, x) = 0$ for every $x \in \Sigma^* \setminus \{ab, abb\}$.
One may observe that $(S\llangle\Sigma^*\rrangle,+,\cdot,0,1)$ is a semiring again.

For $I$ an index set, a family $(r_i ~|~ i \in I)$ of series from $S\llangle\Sigma^*\rrangle$ is \emph{locally finite} if $I(w) = \{i \in I ~|~ (r_i, w) \neq 0\}$
is finite for all $w \in \Sigma^*$. The \emph{sum} over the family $(r_i ~|~ i \in I)$ can then be defined by
\begin{displaymath}
\sum_{i \in I} r_i = r,
\end{displaymath}
where the coefficient $(r,w)$ at each $w \in \Sigma^*$ is given by a \emph{finite} sum
\begin{displaymath}
(r,w) = \sum_{i \in I(w)} (r_i, w).
\end{displaymath}
The \emph{support} of $r \in S\llangle\Sigma^*\rrangle$ is the language 
\begin{displaymath}
\supp(r) = \{w \in \Sigma^* ~|~ (r,w) \neq 0\}.
\end{displaymath}
The~\emph{left quotient} of $r \in S\llangle\Sigma^*\rrangle$ by a word $x \in \Sigma^*$ is a series $x^{-1} r$ such that $(x^{-1}r, w) = (r, xw)$ for~all~$w \in \Sigma^*$.\goodbreak

A \emph{weighted \emph{(}finite\emph{)} automaton} over a semiring $S$ and alphabet $\Sigma$ is a quadruple $\mathcal{A} = (Q,\sigma,\iota,\tau)$, where $Q$ is a finite set of states,
$\sigma\colon Q \times \Sigma \times Q \to S$ a~transition weighting function, $\iota\colon Q \to S$ an initial weighting function, and~$\tau\colon Q \to S$ a terminal weighting function.
We often assume without loss of generality that $Q = [n] = \{1,\ldots,n\}$
for some nonnegative integer $n$; we write $\mathcal{A} = (n,\sigma,\iota,\tau)$ instead of $\mathcal{A} = ([n],\sigma,\iota,\tau)$ in~that~case.\goodbreak

A \emph{transition} of $\mathcal{A} = (Q,\sigma,\iota,\tau)$ is a triple $(p,c,q) \in Q \times \Sigma \times Q$ such that \mbox{$\sigma(p,c,q) \neq 0$}.
A \emph{run} of $\mathcal{A}$ is a word $\gamma = q_0 c_1 q_1 c_2 q_2 \ldots q_{n-1} c_n q_n \in (Q\Sigma)^* Q$, for some $n \in \mathbb{N}$, such that $q_0,\ldots,q_n \in Q$,
$c_1,\ldots,c_n \in \Sigma$, and~$(q_{k-1},c_k,q_k)$ is a transition for $k = 1,\ldots,n$; we also say that $\gamma$ is a~run \emph{from~$q_0$~to~$q_n$} and~that $\gamma$ \emph{passes through} the~states $q_0,\ldots,q_n$.
Moreover, we~write $\lambda(\gamma) = c_1 c_2 \ldots c_n \in \Sigma^*$ for the \emph{label}~of~$\gamma$ and~$\sigma(\gamma) = \sigma(q_0,c_1,q_1) \sigma(q_1,c_2,q_2) \ldots \sigma(q_{n-1},c_n,q_n) \in S$ for the \emph{value} of $\gamma$;
we also say that $\gamma$ is \emph{a~run upon~$\lambda(\gamma)$}.
The~\emph{monomial} $\|\gamma\| \in S\llangle\Sigma^*\rrangle$ realised by the run $\gamma$ is defined by
\begin{displaymath}
\|\gamma\| = \left(\iota(q_0) \sigma(\gamma) \tau(q_n)\right) \lambda(\gamma).
\end{displaymath}
If we denote by $\mathcal{R}(\mathcal{A})$ the set of all runs of the automaton $\mathcal{A}$, then the family of~monomials $(\|\gamma\| ~|~ \gamma \in \mathcal{R}(\mathcal{A}))$ is obviously locally finite
and the \emph{behaviour} of~$\mathcal{A}$ can be defined by the infinite sum
\begin{displaymath}
\|\mathcal{A}\| = \sum_{\gamma \in \mathcal{R}(\mathcal{A})} \|\gamma\|.
\end{displaymath}
In particular, $\|\mathcal{A}\| = 0$ if $Q = \emptyset$.
A series $r \in S\llangle\Sigma^*\rrangle$ is \emph{rational} over $S$ if $r = \|\mathcal{A}\|$ for some weighted automaton $\mathcal{A}$ over $S$ and $\Sigma$.

A~state $q \in Q$ of a weighted automaton $\mathcal{A} = (Q,\sigma,\iota,\tau)$ over $S$ and $\Sigma$ is said to~be \emph{accessible}
if there is a~run in~$\mathcal{A}$ from some $p \in Q$ satisfying $\iota(p) \neq 0$ to $q$.\footnote{Note that the value of this run might be zero in case $S$ is not zero-divisor free.}
Dually, a~state $q \in Q$ is \emph{coaccessible} if there is a run in $\mathcal{A}$ from $q$ to some $p \in Q$ such that $\tau(p) \neq 0$.
The~automaton~$\mathcal{A}$ is~\emph{trim} if all its states are both accessible and coaccessible \cite{sakarovitch2009a}.

Given a~weighted automaton $\mathcal{A} = (Q,\sigma,\iota,\tau)$ and $q \in Q$, we denote by $\|\mathcal{A}\|_q$ the \emph{future} of $q$, \emph{i.e.},
the~series realised by~an~automaton $\mathcal{A}_q = (Q,\sigma,\iota_q,\tau)$ with $\iota_q(q) = 1$ and~$\iota_q(p) = 0$ for all $p \in Q \setminus \{q\}$.
Similarly, we denote by~${}_q\|\mathcal{A}\|$ the~\emph{past} of~$q$, \emph{i.e.}, the~series realised by~an~automaton ${}_q\mathcal{A} = (Q,\sigma,\iota,\tau_q)$ with $\tau_q(q) = 1$ and~$\tau_q(p) = 0$ for all $p \in Q \setminus \{q\}$.    

Let $S^{m \times n}$ be the set of all $m \times n$ matrices over $S$. A \emph{linear representation} of a~weighted automaton $\mathcal{A} = (n,\sigma,\iota,\tau)$ over $S$ and $\Sigma$ 
is given by $\mathcal{P}_{\mathcal{A}} = (n,\mathbf{i},\mu,\mathbf{f})$, where $\mathbf{i} = (\iota(1),\ldots,\iota(n))$,
$\mu\colon (\Sigma^*,\cdot) \to (S^{n \times n},\cdot)$ is a~monoid homomorphism such that for all $c \in \Sigma$ and $i,j \in [n]$, the entry of $\mu(c)$ in~the~$i$-th row and $j$-th column
is given by $\sigma(i,c,j)$, and $\mathbf{f} = (\tau(1),\ldots,\tau(n))^T$. The representation $\mathcal{P}_{\mathcal{A}}$ describes $\mathcal{A}$ unambiguously,
and $(\|\mathcal{A}\|, w) = \mathbf{i}\mu(w)\mathbf{f}$ holds for all $w \in \Sigma^*$.       

As a consequence of this connection to linear representations, methods of linear algebra can be employed in~the~study of~weighted automata \emph{over fields}.
This leads to~a~particularly well-developed theory, including a~polynomial-time minimisation algorithm,
whose basic ideas go back to~M.-P.~Sch\"utzenberger~\cite{schutzenberger1961a}
and~which has been explicitly described by A.~Cardon and~M.~Crochemore~\cite{cardon1980a}. The~reader may consult~\cite{berstel2011a,sakarovitch2009a,sakarovitch2009b} for~a~detailed exposition.

We only note here that the gist of~this minimisation algorithm lies in~an~observation that given a~weighted automaton~$\mathcal{A}$ over a~field~$\mathbb{F}$ and~alphabet~$\Sigma$ with    
$\mathcal{P}_{\mathcal{A}} = (n,\mathbf{i},\mu,\mathbf{f})$, one can find in~polynomial time a~finite language $L = \{x_1,\ldots,x_m\}$ of~words over $\Sigma$ such that $L$ is prefix-closed
and~the~vectors $\mathbf{i}\mu(x_1),\ldots,\mathbf{i}\mu(x_m)$ form a~basis of~the~vector subspace $\mathrm{Left}(\mathcal{A})$ of~$\mathbb{F}^{1 \times n}$ generated by~$\mathbf{i}\mu(x)$ for~$x \in \Sigma^*$.
Such a~language~$L$ is~called a~\emph{left basic language} of~$\mathcal{A}$. Similarly, one can find in~polynomial time
a~\emph{right basic language} of~$\mathcal{A}$ -- \emph{i.e.}, a~finite language $R = \{y_1,\ldots,y_k\}$ of words over $\Sigma$ that is suffix-closed,
and~the~vectors $\mu(y_1)\mathbf{f},\ldots,\mu(y_k)\mathbf{f}$ form a~basis of~the~vector subspace $\mathrm{Right}(\mathcal{A})$ of~$\mathbb{F}^{n \times 1}$ generated by
the~vectors $\mu(y)\mathbf{f}$ for~$y \in \Sigma^*$.

The actual minimisation algorithm then consists of two reduction steps.
The~input automaton~$\mathcal{A}$ with representation $\mathcal{P}_{\mathcal{A}} = (n,\mathbf{i},\mu,\mathbf{f})$ is
first transformed into an~equivalent automaton~$\mathcal{B}$ with $\mathcal{P}_{\mathcal{B}} = (k,\mathbf{i}',\mu',\mathbf{f}')$. Here, $k \leq n$
is the~size of~the~right basic language $R = \{y_1,\ldots,y_k\}$ of~$\mathcal{A}$ with $y_1 = \eps$, 
\begin{equation}
\label{eq:1}
\mathbf{i}' = \mathbf{i}Y, \quad \mu'(c) = Y^{-1}_{\ell} \mu(c) Y \text{ for all $c \in \Sigma$}, \quad \text{ and } \quad \mathbf{f}' = (1,0,\ldots,0)^T,
\end{equation}  
where $Y \in \mathbb{F}^{n \times k}$ is a~matrix of~full column rank with columns $\mu(y_1)\mathbf{f},\ldots,\mu(y_k)\mathbf{f}$ 
and $Y^{-1}_{\ell} \in \mathbb{F}^{k \times n}$ is its left inverse matrix. 
The~automaton~$\mathcal{B}$ is then transformed into a~\emph{minimal} equivalent automaton $\mathcal{C}$ with $\mathcal{P}_{\mathcal{C}} = (m,\mathbf{i}'',\mu'',\mathbf{f}'')$. Here,
$m \leq k$ is the size of~the~left basic language $L = \{x_1,\ldots,x_m\}$ of $\mathcal{B}$ with $x_1 = \eps$,
\begin{equation}
\label{eq:2}
\mathbf{i}'' = (1,0,\ldots,0), \quad \mu''(c) = X \mu'(c) X^{-1}_r \text{ for all $c \in \Sigma$}, \quad \text{ and } \quad \mathbf{f}'' = X \mathbf{f}', 
\end{equation}
where $X \in \mathbb{F}^{m \times k}$ is a~matrix of~full row rank with rows $\mathbf{i}'\mu'(x_1),\ldots,\mathbf{i}'\mu'(x_m)$ and~$X^{-1}_r$ is its right inverse matrix.  
As the vector space $\mathrm{Left}(\mathcal{B})$ -- which is the~row space of~$X$ -- is invariant under $\mu'(c)$ for all $c \in \Sigma$, it follows that  
\begin{equation}
\label{eq:3}
\mathbf{i}'' X = \mathbf{i}', \quad \mu''(c) X = X \mu'(c) \text{ for all $c \in \Sigma$}, \quad \text{ and } \quad \mathbf{f}'' = X \mathbf{f}', 
\end{equation}     
showing that the resulting automaton $\mathcal{C}$ is \emph{conjugate}~\cite{beal2005a,beal2006a} to~$\mathcal{B}$ by~the~matrix~$X$.
Thus $\mathbf{i}'' \mu''(x) X = \mathbf{i}' \mu'(x)$ for all $x \in \Sigma^*$, so that the vector $\mathbf{i}'' \mu''(x)$ represents the~coordinates of~$\mathbf{i}' \mu'(x)$
with respect to the~basis $(\mathbf{i}'\mu'(x_1),\ldots,\mathbf{i}'\mu'(x_m))$ of~the~space $\mathrm{Left}(\mathcal{B})$. In~particular, note that
$(\mathbf{i}''\mu''(x_1),\ldots,\mathbf{i}''\mu''(x_m))$ is the~standard basis of~$\mathbb{F}^m$.       

Finally, let us mention that any weighted automaton~$\mathcal{A}$ over~$\mathbb{F}$ and~$\Sigma$ with representation $\mathcal{P}_{\mathcal{A}} = (n,\mathbf{i},\mu,\mathbf{f})$
gives rise to a~linear mapping $\Lambda[\mathcal{A}]\colon \mathrm{Left}(\mathcal{A}) \to \mathbb{F}\llangle\Sigma^*\rrangle$, uniquely defined by
\begin{equation}
\label{eq:lambda}
\Lambda[\mathcal{A}]\colon \mathbf{i}\mu(x) \mapsto \sum_{w \in \Sigma^*} \left(\mathbf{i}\mu(x)\mu(w)\mathbf{f}\right)\,w = x^{-1} \|\mathcal{A}\|
\end{equation}
for all $x \in \Sigma^*$. This mapping is always injective when $\mathcal{A}$ is a~minimal automaton realising its behaviour~\cite{sakarovitch2009b}.\goodbreak

\section{Bideterministic Weighted Automata over a Semiring}
\label{sec:general}

In the same way as for finite automata without weights \cite{tamm2003a,tamm2004a}, we say that a~weighted automaton~$\mathcal{A}$ is~\emph{bideterministic} if both $\mathcal{A}$ and its transpose are deterministic;
in particular, $\mathcal{A}$ necessarily contains at~most one state with nonzero initial weight and at most one state with nonzero terminal weight.
This is made more precise by the following definition.\goodbreak

\begin{definition}
\label{def:bdwa}
Let $S$ be a semiring and $\Sigma$ an alphabet. A weighted automaton $\mathcal{A} = (Q,\sigma,\iota,\tau)$ over $S$ and $\Sigma$ is \emph{bideterministic}
if all of the following conditions are satisfied:
\begin{enumerate}[label=$(\roman*)$]
\item{There is at most one state $p \in Q$ such that $\iota(p) \neq 0$.
}
\item{If $\sigma(p,c,q) \neq 0$ and $\sigma(p,c,q') \neq 0$ for $p,q,q' \in Q$ and $c \in \Sigma$, then $q = q'$.
}
\item{There is at most one state $q \in Q$ such that $\tau(q) \neq 0$.
}
\item{If $\sigma(p,c,q) \neq 0$ and $\sigma(p',c,q) \neq 0$ for $p,p',q \in Q$ and $c \in \Sigma$, then $p = p'$.
}
\end{enumerate}
\end{definition}

The conditions $(i)$ and $(ii)$ assure that the~automaton~$\mathcal{A}$ is \emph{deterministic}, while the remaining two conditions assure that $\mathcal{A}$ is \emph{codeterministic}.  

It~has~been shown by H.~Tamm and E.~Ukkonen~\cite{tamm2003a,tamm2004a} that
a trim bideterministic automaton without weights is always a~minimal nondeterministic automaton for the~language it recognises. 
Every language recognised by some bideterministic automaton
thus also admits a minimal automaton that is bideterministic. Moreover, by uniqueness of minimal deterministic finite automata and existence of efficient minimisation algorithms,
it is decidable whether a~rational language is recognised by a bideterministic automaton.
 
In what follows, we ask whether these properties generalise to~bideterministic weighted automata over some semiring $S$.
That is, given a~semi\-ring~$S$, we are interested in the following three questions.\footnote{Minimality of an automaton is understood with respect to the number of states in~what follows.}      
\begin{description}
\item[Question 1.]{Is every trim bideterministic weighted automaton over~$S$ necessarily minimal?}
\item[Question 2.]{Does every bideterministic weighted automaton over~$S$ admit an~equivalent minimal weighted automaton over $S$ that is bideterministic?} 
\item[Question 3.]{Is it decidable whether a weighted automaton over~$S$ admits a~bideterministic equivalent?}
\end{description}

An~affirmative answer to Question~1 clearly implies an~affirmative answer to~Question~2 as well.
We~study the first two questions in~Section~\ref{sec:minimal} and~the~last question in~Section~\ref{sec:decid}. 

\section{The Minimality Property of Bideterministic Automata}
\label{sec:minimal}

We now study the conditions on~a~semiring~$S$ under which the~trim bideterministic weighted automata over~$S$
are always minimal, and~answer Question~1, as~well~as the~related Question~2, for three representative classes of semirings.

In~particular, we show that every trim bideterministic weighted automaton over a~\emph{field} -- or, more generally, over an~\emph{integral domain} -- is minimal.
The~same property is observed for~bideterministic weighted automata over \emph{positive semirings}, including for~instance the~\emph{tropical semirings} and~\emph{semirings of~formal languages}.
On~the~other hand, we~observe that the~first question has a~negative answer over every semiring that is not zero-divisor free -- and~thus, in particular, also over
every nontrivial \emph{commutative ring} other than an~integral domain. For~the~second question, we obtain a~negative answer for~a~large class of~commutative semirings, 
which also includes numerous finite commutative rings. Moreover, we completely characterise the~rings $\mathbb{Z}/m\mathbb{Z}$ with $m \in \mathbb{N}$, over which
this question admits an~affirmative answer. 

\subsection{Fields and Integral Domains}

The minimality property of trim bideterministic weighted automata \emph{over fields} follows by the fact that the~Cardon-Crochemore minimisation algorithm for these automata, described
in~Section~\ref{sec:prelim}, preserves both bideterminism and~the~number of~useful states of~a~bideterministic automaton, as we now observe. 

\begin{theorem}
\label{th:car-cro}
Let $\mathcal{A}$ be a bideterministic weighted automaton over a field $\mathbb{F}$. Then the Cardon-Crochemore minimisation algorithm applied to~$\mathcal{A}$ outputs a~bideterministic
weighted automaton $\mathcal{C}$. Moreover, if $\mathcal{A}$ is trim, then $\mathcal{C}$ has the same number of states as $\mathcal{A}$. 
\end{theorem}
\begin{proof}
Let $\mathcal{P} = (n,\mathbf{i},\mu,\mathbf{f})$ be a linear representation of some bideterministic weighted automaton $\mathcal{D}$.
Then there is at most one nonzero entry in each row and column of $\mu(c)$ for each $c \in \Sigma$, and at most one nonzero entry in $\mathbf{i}$ and $\mathbf{f}$.\goodbreak

Moreover, the words $x_1,\ldots,x_m$ of the left basic language of $\mathcal{D}$ correspond bijectively to accessible states of $\mathcal{D}$ and the vector
$\mathbf{i}\mu(x_i)$ contains, for $i = 1,\ldots,m$, exactly one nonzero entry at the position determined by the state corresponding to $x_i$.
Similarly, the words $y_1,\ldots,y_k$ of the right basic language of $\mathcal{D}$ correspond to coaccessible states and the vector $\mu(y_i)\mathbf{f}$ contains,
for $i = 1,\ldots,k$, exactly one nonzero entry. Thus, using these vectors to form the matrices $X$ and~$Y$ as in~Section~\ref{sec:prelim},
we see that one obtains monomial matrices after removing the~zero columns from $X$ and the zero rows from $Y$.
As~a~result, a~right inverse $X_r^{-1}$ of~$X$ can be obtained by taking the reciprocals of all nonzero entries of~$X$ and transposing the resulting matrix,
and similarly for~a~left inverse $Y_{\ell}^{-1}$ of~$Y$. 
  
The matrices $X\mu(c)X_r^{-1}$ and~$Y_{\ell}^{-1}\mu(c)Y$ for $c \in \Sigma^*$ clearly contain at~most one nonzero entry in~each row and column,
and~the~vectors $\mathbf{i} Y$ and~$X\mathbf{f}$ contain at~most one nonzero entry as well. This means that the~reduction step~(\ref{eq:1}) applied to~a~bideterministic
automaton $\mathcal{A}$ yields a~bideterministic automaton~$\mathcal{B}$, and~that the~reduction step~(\ref{eq:2}) applied to~the~bideterministic automaton~$\mathcal{B}$
yields a~bideterministic minimal automaton $\mathcal{C}$ as an output of the algorithm.

When $\mathcal{A}$ is in addition trim, then what has been said implies that the words of~the~right basic language of~$\mathcal{A}$ correspond bijectively to states of~$\mathcal{A}$,
so that the automaton $\mathcal{B}$ obtained via (\ref{eq:1}) has the~same number of~states as $\mathcal{A}$. This automaton is obviously trim as well, 
and~the~words of~the~left basic language of~$\mathcal{B}$ correspond bijectively to states of~$\mathcal{B}$. Hence, the automaton $\mathcal{C}$ obtained via~(\ref{eq:2})
also has the~same number of~states as~$\mathcal{A}$.        
\end{proof}

As every integral domain can be embedded into its field of fractions, the~property established above holds for automata over integral domains as well. 

\begin{corollary}
\label{cor:min-fields}
Every trim bideterministic weighted automaton over an integral domain is minimal.
\end{corollary}

\subsection{Other Commutative Rings}

We now show that the property established above for weighted automata over integral domains \emph{cannot} be generalised to automata over commutative rings.
In fact, we observe that non-minimal trim bideterministic weighted automata do exist over every \emph{semiring} that is not zero-divisor free -- and thus also over
every nontrivial commutative ring that is not an integral domain. Moreover, we exhibit a~class of~commutative semirings~$S$ such that bideterministic weighted automata
over~$S$ do not even always admit a~minimal bideterministic equivalent -- although this class includes many finite commutative rings as well, it includes
neither all nontrivial commutative rings other than integral domains, nor all such finite rings.
\goodbreak

Let us start by a simple observation showing that the answer to Question~1 of Section~\ref{sec:general} is negative whenever the~semiring in~consideration is not zero-divisor free.  

\begin{theorem}
Let $S$ be a semiring that is not zero-divisor free. Then there exists a~trim bideterministic weighted automaton~$\mathcal{A}$ over $S$ that
is not minimal.
\end{theorem}
\begin{proof}
As $S$ is not zero-divisor free, one can find $s,t \in S \setminus \{0\}$ such that $st = 0$. Given such elements, consider the~automaton~$\mathcal{A}$
in~Fig.~\ref{fig:zero-divisors}.

\begin{figure}[h]
\begin{center}
\includegraphics{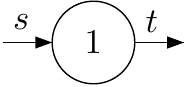}
\end{center}
\caption{\label{fig:zero-divisors}A trim bideterministic weighted automaton $\mathcal{A}$ over $S$ that is not minimal.}
\end{figure}

The automaton $\mathcal{A}$ clearly is both trim and bideterministic. However, it is not minimal: $\|\mathcal{A}\| = 0$, and~the~minimal automaton 
for this series is empty.
\end{proof}

\begin{corollary}
Let $R$ be a commutative ring. Then the following are equivalent:
\begin{enumerate}[label=$(\roman*)$]
\item{Every trim bideterministic weighted automaton~$\mathcal{A}$ over $R$ is minimal. 
}
\item{$R$ is either trivial, or an integral domain.
}
\end{enumerate} 
\end{corollary} 

Classification of commutative rings with affirmative answer to Question~2 of Section~\ref{sec:general} seems to~be less straightforward. We first show that the answer
to this question is negative over a~large class of~commutative semirings including also many finite commutative rings. 

\begin{theorem}
\label{th:commut}
Let $S$ be a commutative semiring containing elements $s,t,s',t' \in S$ such that $st \neq 0 \neq s' t'$ and~$st' = 0 = s't$.
Then there is a~trim bideterministic weighted automaton $\mathcal{A}$ over $S$ such that none of~the~minimal automata
for $\|\mathcal{A}\|$ is bideterministic.  
\end{theorem}
\begin{proof}
Consider a trim bideterministic weighted automaton $\mathcal{A}$ over $S$ depicted in Fig. \ref{fig:1}. 
Clearly, 
\begin{displaymath}
\|\mathcal{A}\| = st \cdot aba + s't' \cdot bb.
\end{displaymath}
The automaton $\mathcal{A}$ is not minimal, as the same series is realised by a smaller automaton $\mathcal{B}$ in Fig. \ref{fig:2}:
\begin{displaymath}
\|\mathcal{B}\| = st \cdot aba + s't' \cdot bb = \|\mathcal{A}\|. 
\end{displaymath}
This gives an alternative way to obtain the negative answer to Question 1 of Section~\ref{sec:general} over $S$.

\begin{figure}[h]
\begin{center}
\includegraphics{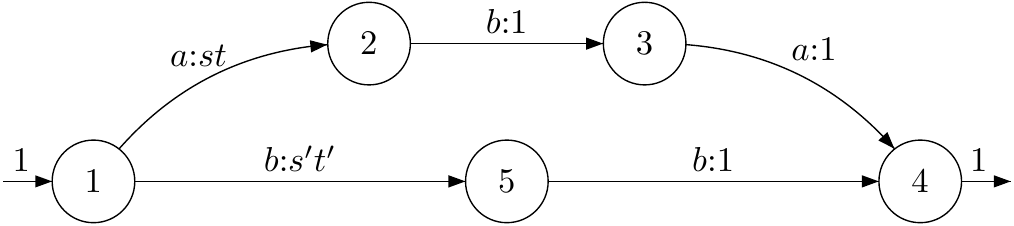}
\end{center}
\caption{\label{fig:1}The trim bideterministic weighted automaton $\mathcal{A}$ over $S$.}
\end{figure}
\begin{figure}[h]
\begin{center}
\includegraphics{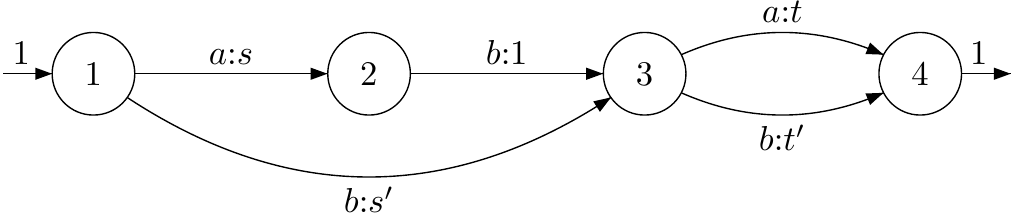}
\end{center}
\caption{\label{fig:2}The four-state weighted automaton $\mathcal{B}$ over $S$ equivalent to $\mathcal{A}$.}
\end{figure}

We show that $\|\mathcal{A}\|$ actually is not realised by any bideterministic weighted automaton over $S$ with less than five states.
This implies that $\mathcal{A}$ is a~counterexample to Question~2 of Section~\ref{sec:general}, and eventually completes the~proof.

Indeed, consider a bideterministic weighted automaton $\mathcal{C} = (Q,\sigma,\iota,\tau)$ such that $\|\mathcal{C}\| = \|\mathcal{A}\|$.
At~least one state with nonzero initial weight is needed to realise $\|\mathcal{A}\|$ by $\mathcal{C}$, as $\|\mathcal{A}\| \neq 0$.
Let us call this state $1$.

As $(\|\mathcal{A}\|,aba) = st \neq 0$, there is a~transition on~$a$ in~$\mathcal{C}$ leading from~$1$. 
This cannot be a~loop at~$1$, as~otherwise $ba$ would have a nonzero coefficient in~$\|\mathcal{C}\|$,
contradicting $\|\mathcal{C}\| = \|\mathcal{A}\|$. It thus leads to some new state, say, $2$.    
 
There has to be a transition on $b$ leading from $2$ and in the same way as~above, we observe that it can lead neither to $1$, nor to $2$, as otherwise
$a$ or $aa$ would have a nonzero coefficient in $\|\mathcal{C}\|$. It thus leads to~some new state~$3$.\goodbreak       

Similar reasoning as above gives us existence of another state $4$, to which a~transition on~$a$ leads from~$3$, and which has
a nonzero terminal weight $\tau(4)$.

Existence of one more state has to be established in order to finish the proof. To this end, observe that $(\|\mathcal{A}\|,bb) = s't' \neq 0$, so that $\mathcal{C}$ has
a transition from $1$ on~$b$, which cannot be a~loop at $1$, as otherwise $b$ would have 
a~nonzero coefficient in~$\|\mathcal{C}\|$. This transition cannot lead to $2$ either, as there already is a~transition~on~$b$ from~$2$ to~$3$, so that $bb$ would have coefficient $0$ in $\|\mathcal{C}\|$. Likewise, 
it cannot lead to~$3$, as there already is a transition
on~$b$ from~$2$ to~$3$ and $\mathcal{C}$ is supposed to be bideterministic. Finally, it also cannot lead to~$4$, as~otherwise there would have to be a~loop labelled by~$b$
at~$4$ and $b$ would have a~nonzero coefficient in~$\|\mathcal{C}\|$. The~transition on~$b$ from $1$ thus indeed leads to some new state $5$. 
\end{proof}

Several natural classes of semirings satisfying the assumptions of Theorem~\ref{th:commut} can be readily identified.

\begin{corollary}
\label{cor:commut1}
Let $S$ be a commutative semiring with elements $u,v \in S$ such that $uv = 0$ and $u^2 \neq 0 \neq v^2$.
Then there is a trim bideterministic weighted automaton $\mathcal{A}$ over $S$ such that none of the minimal automata
for $\|\mathcal{A}\|$ is bideterministic.
\end{corollary}
\begin{proof}
Take $s = t = u$ and $s' = t' = v$ in Theorem \ref{th:commut}.
\end{proof}

Note that the class of commutative semirings from~Corollary~\ref{cor:commut1} also includes many \emph{finite} commutative \emph{rings}.
In particular, the ring $\mathbb{Z}/m\mathbb{Z}$ of integers modulo a~positive integer $m$ falls into this class whenever $m$ has at least two distinct prime factors.

\begin{corollary}
\label{cor:commut2}
Let $m$ be a positive integer with at least two distinct prime factors. 
Then there is a trim bideterministic weighted automaton $\mathcal{A}$ over $\mathbb{Z}/m\mathbb{Z}$, the finite commutative ring of integers modulo $m$, such that none of the minimal automata
for $\|\mathcal{A}\|$ is bideterministic.
\end{corollary}
\begin{proof}
For $p$ a~prime factor of $m$ and $k \in \mathbb{N} \setminus \{0\}$ the~highest exponent such that $p^k$ divides $m$, let $q = m/p^k$. Then the ring $\mathbb{Z}/m\mathbb{Z}$ satisfies the assumptions of Corollary~\ref{cor:commut1} with $u = p^k$ and $v = q$.  
\end{proof}

Let us mention one more particular class of semirings satisfying the assumptions of~Theorem~\ref{th:commut}. This class includes, e.g., 
the commutative ring $\mathbb{Q}[x,y]/(x^2, y^2)$. 

\begin{corollary}
\label{cor:commut3}
Let $S$ be a commutative semiring with elements $u,v \in S$ such that $uv \neq 0$ and $u^2 = 0 = v^2$.
Then there is a trim bideterministic weighted automaton $\mathcal{A}$ over $S$ such that none of the minimal automata
for $\|\mathcal{A}\|$ is bideterministic.
\end{corollary}
\begin{proof}
Take $s = t' = u$ and $t = s' = v$ in Theorem \ref{th:commut}.
\end{proof}\goodbreak        

\begin{remark}
The class of commutative semirings satisfying Theorem~\ref{th:commut} can be conveniently described in~terms of~zero-divisor graphs~\cite{anderson2011a,anderson1999a,demeyer2005a,dolzan2012a} --
more precisely, in~terms of~their looped version appearing, e.g., in~\cite{anderson2011a}. Let $S$ be a~commutative semiring and~$Z(S)^*$ its set of~nontrivial divisors of~zero, \emph{i.e.},
\begin{displaymath}
Z(S)^* = Z(S) \setminus \{0\}
\end{displaymath}
for
\begin{displaymath}
Z(S) = \{a \in S \mid \exists b \in S \setminus \{0\}: ab = 0\}.
\end{displaymath}
The \emph{zero-divisor graph} of~$S$ then is a~possibly infinite undirected graph with $Z(S)^*$ 
as its set of~vertices such that vertices $a,b \in Z(S)^*$ are connected by an~edge (a~loop in~case $a = b$) if and only if $ab = 0$.

A~commutative semiring $S$ satisfies the~assumptions of~Theorem~\ref{th:commut} if and~only if its zero-divisor graph contains not necessarily distinct vertices $s,t,s',t'$ such
that there are edges connecting $s$ with~$t'$ and~$s'$ with~$t$, while there are no edges connecting $s$ with~$t$ and~$s'$ with~$t'$. This ``configuration'' is illustrated in~Fig.~\ref{fig:zd-graph}.      
\begin{figure}[h!]
\begin{center}
\includegraphics{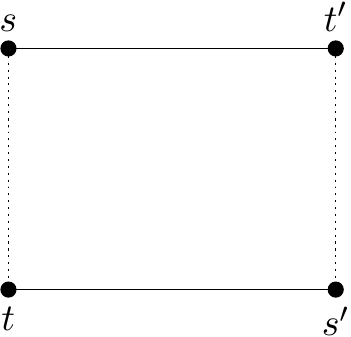}
\end{center}
\caption{\label{fig:zd-graph}A~``configuration'' that has to be contained in~the~zero-divisor graph of~every commutative semiring satisfying the~assumptions of~Theorem~\ref{th:commut}. Full lines in~the~diagram represent edges, while dotted lines should be interpreted as nonexistent edges. 
The~vertices $s,t,s',t'$ might not be distinct, so a~line between two vertices may also indicate (non)existence of~a~loop.}
\end{figure}
\end{remark}

\begin{remark}
Theorem~\ref{th:commut} significantly narrows the~class of~commutative semirings, for~which Question~2 of~Section~\ref{sec:general} may admit an~affirmative answer.
Although we show below that this actually also happens for~some finite commutative rings other than integral domains, Theorem~\ref{th:commut} implies that
some special properties are needed for~this to be the~case. More precisely, the~set of~zero divisors $Z(S)$ in~a~commutative semiring~$S$ is obviously always closed under multiplication
by~an~arbitrary element of $S$, but it may have little additive structure in~general -- see, e.g., \cite{anderson2011a} for~the~case of~rings.
It follows by Theorem~\ref{th:commut} that $Z(S)$ is an~ideal of~$S$ whenever $S$ is a~commutative semiring such that all bideterministic weighted automata over~$S$
admit an~equivalent minimal automaton that is bideterministic as well. Indeed, if $a,b \in Z(S)$ are elements of~$S$ such that $a + b \not\in Z(S)$
and~$c,d \in S \setminus \{0\}$ satisfy $ac = bd = 0$, then 
\begin{displaymath}
0 \neq (a + b) c = ac + bc = bc
\end{displaymath}
and
\begin{displaymath}
0 \neq (a + b) d = ad + bd = ad.
\end{displaymath}
Thus, by~applying Theorem~\ref{th:commut} with $s = a$, $t = d$, $s' = b$, and $t' = c$, we find out that there is a~bideterministic weighted automaton over~$S$
such that none of~the~equivalent minimal automata is bideterministic. 

The condition of~$Z(S)$ forming an~ideal is thus necessary in~order for~Question~2 of~Section~\ref{sec:general} to have an~affirmative answer over a~commutative semiring~$S$.
It is not a~sufficient condition, as can be exemplified by~the~commutative ring $R = \mathbb{Q}[x,y]/(x^2, y^2)$, which satisfies the~conditions of~Corollary~\ref{cor:commut3} although
$Z(R)$ clearly is an~ideal of~$R$. 
\end{remark}\goodbreak          

We leave open the full characterisation of commutative rings, over which bideterministic weighted automata always admit minimal bideterministic equivalents.
However, in what follows we at~least characterise the~rings $\mathbb{Z}/m\mathbb{Z}$ with this property. Recall that the~answer to~Question~2 of~Section~\ref{sec:general}
is affirmative over fields and~integral domains; this includes all the~fields $\mathbb{Z}/m\mathbb{Z}$ with $m$ prime and the integral domain $\mathbb{Z} \cong \mathbb{Z}/0\mathbb{Z}$.
The~answer is also clearly affirmative over the~trivial ring $\mathbb{Z}/1\mathbb{Z}$.
On~the~other hand, the~answer to~the~said question is negative by~Corollary~\ref{cor:commut2} when $m$ contains at~least two distinct prime factors. The~only case left unexplored thus is the one with $m = p^k$ for $p$ prime
and~$k \geq 2$ an~integer. We now show that the~answer to~Question~2 of~Section~\ref{sec:general} is affirmative for all rings $\mathbb{Z}/m\mathbb{Z}$ with $m$ taking this form.

In order to obtain this result, we need the following lemma about linear combinations in the module $(\mathbb{Z}/p^k \mathbb{Z})^n$ over the ring $\mathbb{Z}/p^k \mathbb{Z}$.     

\begin{lemma}
\label{le:lin-combinations}
Let $p$ be a prime, $k,n \in \mathbb{N} \setminus \{0\}$, and $\mathbf{v}_1,\ldots,\mathbf{v}_{n-1}$ vectors in $(\mathbb{Z}/p^k \mathbb{Z})^n$. 
Then there exists at least one index $j \in [n]$, for which there are no $a_1,\ldots,a_{n-1} \in \mathbb{Z}/p^k \mathbb{Z}$ and $b \in (\mathbb{Z}/p^k \mathbb{Z}) \setminus \{0\}$ such that
\begin{displaymath}
a_1 \mathbf{v}_1 + \ldots + a_{n-1} \mathbf{v}_{n-1} = b\mathbf{e}_j,  
\end{displaymath}
where 
\begin{displaymath}
\mathbf{e}_j = (\underbrace{0, \ldots, 0}_{j - 1}, 1, \underbrace{0, \ldots, 0}_{n - j}).
\end{displaymath}
\end{lemma}
\begin{proof}
Suppose for contradiction that there is a prime number $p$ and $k,n \in \mathbb{N} \setminus \{0\}$, for which the~statement of~the~lemma does not hold, \emph{i.e.},
one can find vectors $\mathbf{v}_1,\ldots,\mathbf{v}_{n-1} \in (\mathbb{Z}/p^k \mathbb{Z})^n$ such that a~nonzero multiple of $\mathbf{e}_j$ belongs to
\begin{displaymath}
\langle \mathbf{v}_1,\ldots,\mathbf{v}_{n-1}\rangle = \{a_1 \mathbf{v}_1 + \ldots + a_{n-1} \mathbf{v}_{n-1} \mid a_1,\ldots,a_{n-1} \in \mathbb{Z}/p^k \mathbb{Z}\}
\end{displaymath}
for $j = 1,\ldots,n$. Note that since $p^{k-1}$ is a multiple of every nonzero element of $\mathbb{Z}/p^k \mathbb{Z}$, we actually get
\begin{equation}
\label{lemma:eq1}
p^{k-1} \mathbf{e}_1,\ldots,p^{k-1} \mathbf{e}_{n} \in \langle \mathbf{v}_1,\ldots,\mathbf{v}_{n-1}\rangle.  
\end{equation} 
This is impossible for $n = 1$, so we may assume for the rest of the proof that $n \geq 2$. 

We now observe that with some care, a variant of Gaussian elimination can be applied to the vectors $\mathbf{v}_1,\ldots,\mathbf{v}_{n-1}$.
Take $i \in [n-1]$ and $j \in [n]$ such that the highest power of~$p$ dividing the $j$-th entry of~$\mathbf{v}_i$ is minimised -- then all entries of the vectors $\mathbf{v}_1,\ldots,\mathbf{v}_{n-1}$ 
are multiples of~the~$j$-th entry of~$\mathbf{v}_i$ over~$\mathbb{Z}/p^k \mathbb{Z}$.
This $j$-th entry of $\mathbf{v}_i$ has to be nonzero, as otherwise all of~the~vectors $\mathbf{v}_1,\ldots,\mathbf{v}_{n-1}$ would be zero,
contradicting~(\ref{lemma:eq1}). Moreover, we may assume without loss of generality that actually $j = 1$, as we can relabel the~indices when this is not the~case.    
Then, by~subtracting suitable multiples of~$\mathbf{v}_i$ from the~remaining vectors and switching the~first and~the~$i$-th vector,
we obtain vectors \smash{$\mathbf{v}^{(1)}_1,\ldots,\mathbf{v}^{(1)}_{n-1}$} such that 
\begin{equation}
\label{lemma:eq2}
p^{k-1} \mathbf{e}_1,\ldots,p^{k-1} \mathbf{e}_{n} \in \left\langle \mathbf{v}_1^{(1)},\ldots,\mathbf{v}_{n-1}^{(1)}\right\rangle
\end{equation} 
and~\smash{$\mathbf{v}^{(1)}_1$} is the~only vector with nonzero entry in~the~first column. Moreover, all entries of~these vectors are multiples of~the~first entry of~\smash{$\mathbf{v}^{(1)}_1$} over $\mathbb{Z}/p^k \mathbb{Z}$, so~that $p^{k-1}\mathbf{e}_1$ is the~only among the vectors $p^{k-1} \mathbf{e}_1,\ldots,p^{k-1} \mathbf{e}_{n}$
that can be expressed by~a~linear combination
\begin{displaymath}
a_1 \mathbf{v}^{(1)}_1 + \ldots + a_{n - 1} \mathbf{v}^{(1)}_{n - 1}  
\end{displaymath}
of~the~vectors \smash{$\mathbf{v}^{(1)}_1,\ldots,\mathbf{v}^{(1)}_{n-1}$} with $a_1,\ldots,a_{n-1} \in \mathbb{Z}/p^k \mathbb{Z}$ such that \smash{$a_1 \mathbf{v}^{(1)}_1$} is \emph{not} equal to the~zero vector.
We~may thus assume that actually \smash{$\mathbf{v}^{(1)}_1 = p^{k-1} \mathbf{e}_1$}, while (\ref{lemma:eq2}) remains valid.\goodbreak

The remaining vectors can be dealt with similarly. For $\ell \in [n-2]$, assume that \smash{$\mathbf{v}^{(\ell)}_1,\ldots,\mathbf{v}^{(\ell)}_{n-1}$} are vectors such that
\begin{displaymath}
p^{k-1} \mathbf{e}_1,\ldots,p^{k-1} \mathbf{e}_{n} \in \left\langle \mathbf{v}_1^{(\ell)},\ldots,\mathbf{v}_{n-1}^{(\ell)}\right\rangle,
\end{displaymath}
the entries of $\mathbf{v}_{\ell + 1}^{(\ell)},\ldots,\mathbf{v}_{n-1}^{(\ell)}$ in~the~first $\ell$ columns are all zero, and
\begin{displaymath}
\mathbf{v}_1^{(\ell)} = p^{k-1} \mathbf{e}_1, \qquad \ldots, \qquad \mathbf{v}_{\ell}^{(\ell)} = p^{k-1} \mathbf{e}_{\ell}. 
\end{displaymath}
One can then find $i \in \{\ell + 1, \ldots, n - 1\}$ and $j \in \{\ell + 1, \ldots, n\}$ such that all entries of~the~vectors \smash{$\mathbf{v}_{\ell + 1}^{(\ell)},\ldots,\mathbf{v}_{n-1}^{(\ell)}$} in columns $\ell + 1,\ldots,n$
are multiples of~the~nonzero $j$-th entry of~\smash{$\mathbf{v}_i^{(\ell)}$} over $\mathbb{Z}/p^k \mathbb{Z}$. Without loss of~generality, we may assume that $j = \ell + 1$.
Then, after subtracting suitable multiples of~\smash{$\mathbf{v}^{(\ell)}_i$} from the~remaining among the~vectors \smash{$\mathbf{v}_{\ell + 1}^{(\ell)},\ldots,\mathbf{v}_{n-1}^{(\ell)}$} and~switching the~$(\ell + 1)$-th and~the~$i$-th vector,
we obtain vectors \smash{$\mathbf{v}^{(\ell + 1)}_1,\ldots,\mathbf{v}^{(\ell + 1)}_{n-1}$} such that  
\begin{displaymath}
p^{k-1} \mathbf{e}_1,\ldots,p^{k-1} \mathbf{e}_{n} \in \left\langle \mathbf{v}_1^{(\ell+1)},\ldots,\mathbf{v}_{n-1}^{(\ell+1)}\right\rangle
\end{displaymath}
and the entries of $\mathbf{v}_{\ell + 2}^{(\ell + 1)},\ldots,\mathbf{v}_{n-1}^{(\ell + 1)}$ in~the~first $\ell + 1$ columns are all zero. Moreover, similarly as above, we may
assume \smash{$\mathbf{v}_{\ell + 1}^{(\ell+1)} = p^{k-1} \mathbf{e}_{\ell+1}$}, so that  
\begin{displaymath}
\mathbf{v}_1^{(\ell+1)} = p^{k-1} \mathbf{e}_1, \qquad \ldots, \qquad \mathbf{v}_{\ell + 1}^{(\ell+1)} = p^{k-1} \mathbf{e}_{\ell+1}. 
\end{displaymath}
Thus, in the end of this process, we are finally left with vectors 
\begin{displaymath}
\mathbf{v}^{(n-1)}_1 = p^{k-1} \mathbf{e}_1, \qquad \ldots, \qquad \mathbf{v}^{(n-1)}_{n-1} = p^{k-1} \mathbf{e}_{n-1},  
\end{displaymath}
while at the same time
\begin{displaymath}
p^{k-1} \mathbf{e}_{n} \in \left\langle \mathbf{v}_1^{(n-1)},\ldots,\mathbf{v}_{n-1}^{(n-1)}\right\rangle,
\end{displaymath}
a contradiction.
\end{proof}

We are now prepared to give a positive answer to Question~2 of~Section~\ref{sec:general} for automata over the rings $\mathbb{Z}/p^k\mathbb{Z}$, where $p$ is a~prime and~$k \geq 1$ is an~integer.
This is based on~an~observation that although trim bideterministic automata over these rings might not be minimal when $k \geq 2$, the~minimality property does hold for a~class of bideterministic
weighted automata that are trim in~a~slightly stronger sense made precise by the following definition. 

\begin{definition}
A weighted automaton $\mathcal{A} = (Q,\sigma,\iota,\tau)$ over a~semiring~$S$ and~alphabet~$\Sigma$ is \emph{semantically trim} if for each state $q \in Q$, there exists a~run $\gamma_q$ of~$\mathcal{A}$
passing through~$q$ such that $\|\gamma_q\| \neq 0$.
\end{definition}

Every weighted automaton $\mathcal{A}$ can be transformed to an~equivalent semantically trim automaton by~simply removing the~states~$q$ for which there is no run $\gamma_q$ of~$\mathcal{A}$
realising a~nonzero monomial passing through~$q$. Moreover, every semantically trim automaton is obviously trim, while the~converse might not hold if the~underlying semiring~$S$ is not zero-divisor free.  

\begin{theorem}
\label{th:sem-trim}
Let $p$ be a prime, $k \in \mathbb{N}$, and~$\mathcal{A}$ a~semantically trim bideterministic weighted automaton over $\mathbb{Z}/p^k\mathbb{Z}$. Then $\mathcal{A}$ is minimal. 
\end{theorem}
\begin{proof}
The theorem is trivially true for $k = 0$; we may thus assume that $k \geq 1$. Let $\mathcal{A} = (n,\sigma,\iota,\tau)$ be a~semantically trim bideterministic
weighted automaton over $\mathbb{Z}/p^k\mathbb{Z}$ and~alphabet~$\Sigma$. The automaton is minimal if $n = 0$, so we may assume that $n \geq 1$. As $\mathcal{A}$ is semantically trim, there are words $u_1,\ldots,u_n,v_1,\ldots,v_n \in \Sigma^*$ such that
\begin{equation}
\label{eq:nonzero}
({}_q\|\mathcal{A}\|,u_q) \neq 0, \qquad (\|\mathcal{A}\|_q,v_q) \neq 0, \qquad \text{ and } \qquad (\|\mathcal{A}\|, u_q v_q) \neq 0
\end{equation} 
for $q = 1,\ldots,n$. In particular, this implies that there is precisely one state $q_i \in [n]$ with $\iota(q_i) \neq 0$ and~precisely one state $q_t \in [n]$ with $\tau(q_t) \neq 0$,
and that, for $q = 1,\ldots,n$, there are runs upon $u_q$ from $q_i$ to $q$ and upon $v_q$ from $q$ to $q_t$ in $\mathcal{A}$. Moreover, it follows by bideterminism of~$\mathcal{A}$
that there cannot be a~run in~$\mathcal{A}$ upon~$v_q$ leading from~$p$ to~$q_t$ if $p$ and~$q$ are distinct states of~$\mathcal{A}$. As a~consequence,
\begin{displaymath}
(\|\mathcal{A}\|, u_p v_q) \neq 0
\end{displaymath}
holds for $p,q \in [n]$ if and only if $p = q$.\goodbreak

Now, suppose for contradiction that $\mathcal{A}$ is not minimal. By possibly adding some useless states to a~minimal equivalent automaton, we find out that
there is a~weighted automaton $\mathcal{B} = (n - 1, \sigma', \iota', \tau')$ such that $\|\mathcal{B}\| = \|\mathcal{A}\|$. For $p = 1,\ldots,n-1$, let
\begin{displaymath}
\mathbf{v}_p = \left((\|\mathcal{B}\|_p, v_1),\ldots,(\|\mathcal{B}\|_p, v_n)\right).
\end{displaymath} 
Then, for $q = 1,\ldots,n$,
\begin{displaymath}
\sum_{p = 1}^{n - 1} ({}_p\|\mathcal{B}\|, u_q) \mathbf{v}_p = \left((\|\mathcal{A}\|, u_q v_1), \ldots, (\|\mathcal{A}\|, u_q v_n)\right) = (\|\mathcal{A}\|, u_q v_q) \mathbf{e}_q,
\end{displaymath}
where
\begin{displaymath}
\mathbf{e}_q = (\underbrace{0, \ldots, 0}_{q - 1}, 1, \underbrace{0, \ldots, 0}_{n - q}).
\end{displaymath}
As $(\|\mathcal{A}\|, u_q v_q) \in (\mathbb{Z}/p^k \mathbb{Z}) \setminus \{0\}$ by (\ref{eq:nonzero}), this contradicts Lemma~\ref{le:lin-combinations}.
\end{proof}

\begin{corollary}
\label{cor:mod-rings}
Let $p$ be a prime and~$k \in \mathbb{N}$. Then every bideterministic weighted automaton over $\mathbb{Z}/p^k\mathbb{Z}$ admits an~equivalent minimal automaton that is bideterministic.
\end{corollary}
\begin{proof}
Given a bideterministic weighted automaton $\mathcal{A}$ over $\mathbb{Z}/p^k\mathbb{Z}$, one can obtain a~semantically trim weighted automaton $\mathcal{B}$ equivalent to $\mathcal{A}$ by possibly
removing several states. This automaton $\mathcal{B}$ has to be bideterministic as well -- it is thus minimal by~Theorem~\ref{th:sem-trim}.  
\end{proof}

The result just established finishes the characterisation of rings $\mathbb{Z}/m\mathbb{Z}$, over which all bideterministic weighted automata admit an~equivalent minimal automaton that is bideterministic as well.
This is summarised by~the~following corollary.

\begin{corollary}
Let $m \in \mathbb{N}$. Then the following are equivalent:
\begin{enumerate}[label=$(\roman*)$]
\item{For all bideterministic weighted automata $\mathcal{A}$ over $\mathbb{Z}/m\mathbb{Z}$, there exists an~equivalent minimal automaton that is bideterministic.
}
\item{Either $m = 0$, or $m = p^k$ for $p$ a~prime and~$k \in \mathbb{N}$.
}
\end{enumerate}
\end{corollary}
\begin{proof}
If $m = 0$, then $\mathbb{Z}/m\mathbb{Z} \cong \mathbb{Z}$, and~existence of~minimal bideterministic equivalents to all bideterministic automata follows
by~Corollary~\ref{cor:min-fields}. For $m = p^k$ with $p$ prime and $k \in \mathbb{N}$, the~same property follows by~Corollary~\ref{cor:mod-rings} (and~for $k = 1$ also by~Corollary~\ref{cor:min-fields}).
On the other hand, if $m \geq 2$ is not equal~to~$p^k$ for~$p$ prime and~$k \in \mathbb{N}$, then $m$ necessarily has at least two different prime factors, 
and~Corollary~\ref{cor:commut2} implies existence of a~bideterministic weighted automaton over $\mathbb{Z}/m\mathbb{Z}$ that admits no minimal bideterministic equivalent.    
\end{proof}
      
\subsection{Positive Semirings}
\label{sec:positive}

We now observe that the minimality property does hold for trim bideterministic weighted automata \emph{over positive semirings}.
Recall that a semiring is \emph{positive} if it is both zero-sum free and zero-divisor free. This class includes for instance the~\emph{tropical semirings},
\emph{semirings of formal languages}, and the \emph{Boolean semiring}. 

\begin{theorem}
Every trim bideterministic weighted automaton over a positive semiring is minimal.
\end{theorem}
\begin{proof}
Let $\mathcal{A}$ be a~trim bideterministic weighted automaton over a positive semi\-ring $S$.
By~positivity of~$S$, the~language $\supp(\|\mathcal{A}\|)$ is recognised by~a~trim bideterministic finite automaton $\mathcal{A}'$ obtained from~the~automaton~$\mathcal{A}$
by~``forgetting about its weights''. This is a~minimal nondeterministic automaton for~$\supp(\|\mathcal{A}\|)$ by~the~minimality property
of~trim bideterministic automata without weights~\cite{tamm2003a,tamm2004a}. 

Now, if $\mathcal{A}$ was not minimal, there would be a~smaller weighted automaton $\mathcal{B}$ over~$S$ such that $\|\mathcal{B}\| = \|\mathcal{A}\|$.
By~``forgetting about its weights'', we would obtain a~nondeterministic finite automaton $\mathcal{B}'$ recognising $\supp(\|\mathcal{B}\|) = \supp(\|\mathcal{A}\|)$.
However, $\mathcal{B}'$ is smaller than $\mathcal{A}'$, contradicting the~minimality of~$\mathcal{A}'$.
\end{proof}

\section{Decidability of Bideterminisability}
\label{sec:decid}

Let us now consider the problem of deciding whether a given weighted automaton admits a bideterministic equivalent. 
We show that this~\emph{bideterminisability} problem is decidable both over effective fields and over tropical semirings (of nonnegative integers, integers, and~rationals).
Moreover, the~decision procedure over fields is simply based on checking whether the output of the Cardon-Crochemore minimisation algorithm is bideterministic or not -- it thus runs in~polynomial time.

By contrast, the results known so far about the~\emph{determinisability} problem are weaker: its decidability status over tropical semirings is still open~\cite{lombardy2021a,lombardy2006a},
and~its decidability over fields has only recently been established by J.~P.~Bell and~D.~Smertnig~\cite{bell2023a}. Moreover, there are no reasonable complexity bounds known 
for~the~general determinisability problem for weighted automata over fields, and a polynomial-time algorithm is known just for unary weighted automata over the rationals \cite{kostolanyi2022a}. 

\subsection{Fields and Integral Domains}

We first prove decidability of the bideterminisability problem for automata over fields. To this end, we strengthen
Theorem~\ref{th:car-cro} by showing that the Cardon-Crochemore minimisation algorithm outputs a~bideterministic automaton not only when applied to a bideterministic automaton,
but also when applied to any bideterminisable automaton. To decide bideterminisability, it thus suffices to run this algorithm and find out whether its output is bideterministic.

\begin{lemma}
\label{le:min-det}
Let $\mathcal{A}$ be a weighted automaton over a field $\mathbb{F}$ such that some of the minimal automata equivalent to $\mathcal{A}$ is deterministic.
Then the Cardon-Crochemore algorithm applied to $\mathcal{A}$ outputs a deterministic automaton.
\end{lemma}
\begin{proof}
Let $\mathcal{C}$ with $\mathcal{P}_{\mathcal{C}} = (m,\mathbf{i},\mu,\mathbf{f})$ be an automaton obtained as an output of the Cardon-Crochemore algorithm for input automaton $\mathcal{A}$,
and $L = \{x_1,\ldots,x_m\}$ with $x_1 = \eps$ the left basic language used in~reduction step~(\ref{eq:2}).
Then $\mathbf{i}\mu(x)$ represents, for all $x \in \Sigma^*$, the~coordinates of~the~series $x^{-1}\|\mathcal{A}\|$ with respect to~the~basis
$(x_1^{-1}\|\mathcal{A}\|,\ldots,x_m^{-1}\|\mathcal{A}\|)$ of~the~vector space $\mathcal{Q}(\|\mathcal{A}\|)$ generated by left quotients of $\|\mathcal{A}\|$ by words.\goodbreak

To see this, recall that $(\mathbf{i}\mu(x_1),\ldots,\mathbf{i}\mu(x_m))$ is the~standard basis of~$\mathbb{F}^m$
and~that the~linear mapping $\Lambda[\mathcal{C}]$ given as in~(\ref{eq:lambda}) is injective by~minimality of~$\mathcal{C}$.
As~the~image of~$\Lambda[\mathcal{C}]$ spans $\mathcal{Q}(\|\mathcal{C}\|) = \mathcal{Q}(\|\mathcal{A}\|)$, we see that
\begin{displaymath}
\left(x_1^{-1}\|\mathcal{A}\|,\ldots,x_m^{-1}\|\mathcal{A}\|\right) = \left(\Lambda[\mathcal{C}](\mathbf{i}\mu(x_1)),\ldots,\Lambda[\mathcal{C}](\mathbf{i}\mu(x_m))\right)
\end{displaymath}
is indeed a~basis of~$\mathcal{Q}(\|\mathcal{A}\|)$.
Moreover, given an~arbitrary word $x \in \Sigma^*$ with $\mathbf{i}\mu(x) = (a_1,\ldots,a_m) \in \mathbb{F}^m$, we obtain
\begin{align*}
x^{-1}\|\mathcal{A}\| & = \Lambda[\mathcal{C}](\mathbf{i}\mu(x)) = \Lambda[\mathcal{C}](a_1 \mathbf{i}\mu(x_1) + \ldots + a_m \mathbf{i}\mu(x_m)) = \\
& = a_1 \Lambda[\mathcal{C}](\mathbf{i}\mu(x_1)) + \ldots + a_m \Lambda[\mathcal{C}](\mathbf{i}\mu(x_m)) = \\
& = a_1 x_1^{-1} \|\mathcal{A}\| + \ldots + a_m x_m^{-1} \|\mathcal{A}\|, 
\end{align*} 
from which the~said property follows.

Now, assume for contradiction that $\mathcal{C}$ is not deterministic. By minimality of~$\mathcal{C}$, there is some $x \in \Sigma^*$ such that $\mathbf{i}\mu(x)$ contains
at least two nonzero entries. However, by our assumptions, there also is an $m$-state \emph{deterministic} automaton $\mathcal{D}$ such that $\|\mathcal{D}\| = \|\mathcal{A}\|$.
Linear independence of $x_1^{-1}\|\mathcal{A}\|,\ldots,x_m^{-1}\|\mathcal{A}\|$ implies that the~$m$ states of~$\mathcal{D}$ can be labelled as $q_1,\ldots,q_m$ so that
$x_i^{-1}\|\mathcal{A}\|$ is a~scalar multiple of~$\|\mathcal{D}\|_{q_i}$ for~$i = 1,\ldots,m$. By determinism of~$\mathcal{D}$, every $x^{-1}\|\mathcal{A}\|$ with $x \in \Sigma^*$
is a scalar multiple of some $\|\mathcal{D}\|_{q_i}$ with $i \in [m]$, and~hence also of~some $x_i^{-1}\|\mathcal{A}\|$. It thus follows that there is some $x \in \Sigma^*$ such
that $x^{-1}\|\mathcal{A}\|$ has two different coordinates with respect to $(x_1^{-1}\|\mathcal{A}\|,\ldots,x_m^{-1}\|\mathcal{A}\|)$: a contradiction.    
\end{proof}

\begin{theorem}
Let $\mathcal{A}$ be a weighted automaton over a field. If $\mathcal{A}$ has a~bideterministic equivalent,
then the~Cardon-Crochemore algorithm applied to $\mathcal{A}$ outputs a~bideterministic automaton.
\end{theorem}\goodbreak
\begin{proof}
Let $\mathcal{A}$ admit a bideterministic equivalent $\mathcal{B}$, and assume that the~automaton~$\mathcal{B}$ is trim. Then $\mathcal{B}$ is minimal by~Corollary~\ref{cor:min-fields},
so Lemma~\ref{le:min-det} implies that the algorithm applied to $\mathcal{A}$ yields a deterministic automaton $\mathcal{D}$.
If $\mathcal{D}$ was not bideterministic, then there would be $u,v \in \Sigma^*$ such that $u^{-1}\|\mathcal{D}\|$ is not a~scalar multiple of~$v^{-1}\|\mathcal{D}\|$ and~$\supp(u^{-1}\|\mathcal{D}\|) \cap \supp(v^{-1}\|\mathcal{D}\|) \neq \emptyset$.
On~the~other hand, bideterminism of~$\mathcal{B}$ implies\footnote{This is a straightforward extension of a well-known property of bideterministic finite automata without weights -- see, e.g., L.~Pol\'ak~\cite[Section 5]{polak2005a}.} $\supp(u^{-1}\|\mathcal{B}\|) \cap \supp(v^{-1}\|\mathcal{B}\|) = \emptyset$
when $u^{-1}\|\mathcal{B}\|$ is not a~scalar multiple of~$v^{-1}\|\mathcal{B}\|$.
This contradicts the~assumption that $\|\mathcal{B}\| = \|\mathcal{D}\| = \|\mathcal{A}\|$. 
\end{proof}

\begin{corollary}
Bideterminisability of weighted automata over effective fields is decidable in polynomial time.
\end{corollary}

We leave the decidability status of the bideterminisability problem open for weighted automata over integral domains. 
However, the following two examples indicate that the~simple decision procedure for~automata over fields is no longer sufficient when weights are taken from an~integral domain.

In particular, Example \ref{ex:int-domains1} shows that a~weighted automaton over an~integral domain~$R$ can be bideterminisable over the fraction field of~$R$ even if it does not admit
a~bideterministic equivalent over $R$ itself. The~Cardon-Crochemore minimisation algorithm thus may produce 
a~bideterministic automaton $\mathcal{B}$ over the~field of~fractions of~$R$ even if the~input automaton $\mathcal{A}$
is not bideterminisable over $R$, so that simply checking whether $\mathcal{B}$ is bideterministic is insufficient to decide bideterminisability of $\mathcal{A}$.

On~the~other hand, in Example \ref{ex:int-domains2} we observe that although the~Cardon-Crochemore algorithm applied to~a~bideterminisable automaton over an~integral domain~$R$
always outputs a~bideterministic automaton over the~fraction field of $R$, it may not output a~bideterministic automaton over~$R$ itself. This means that running the~Cardon-Crochemore
algorithm and~checking whether the~result is a~bideterministic automaton \emph{over $R$} is not appropriate either.      

\begin{example}
\label{ex:int-domains1}
Consider the subring of the integral domain $\mathbb{Q}[x,y]$ generated over $\mathbb{Q}$ by the set of all monomials of degree at least two, \emph{i.e.},
the integral domain $\mathbb{Q}[X] \subseteq \mathbb{Q}[x,y]$ for~$X = \{x^m y^n \mid m + n \geq 2\}$.
As~one can obtain both $x$~and~$y$ as~quotients of~monomials from~$X$ -- e.g., $x = x^3/x^2$ and $y = y^3/y^2$ -- the~fraction field of~$\mathbb{Q}[X]$ equals~$\mathbb{Q}(x,y)$.
 
Let $\mathcal{A}$ be a weighted automaton over the integral domain $\mathbb{Q}[X]$ and alphabet $\Sigma = \{a\}$ given in Fig. \ref{fig:id1}.
Clearly,
\begin{displaymath}
\|\mathcal{A}\| = \sum_{t \in \mathbb{N}} (x + y)^{t + 2}\,a^t.
\end{displaymath}
 
\begin{figure}[h!]
\begin{center}
\includegraphics{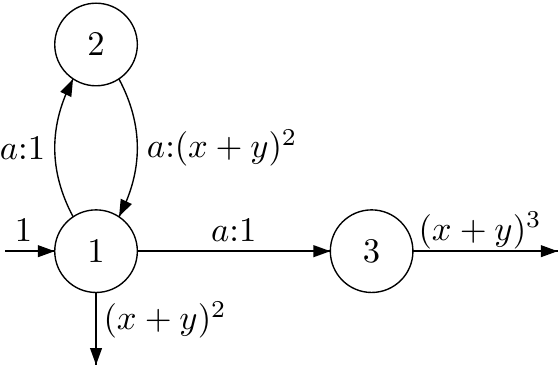}
\end{center}
\caption{\label{fig:id1}The weighted automaton $\mathcal{A}$ over $\mathbb{Q}[X]$ and $\Sigma = \{a\}$.}
\end{figure}

It is easy to see that $\mathcal{A}$ is equivalent to the bideterministic weighted automaton $\mathcal{B}$ over $\mathbb{Q}(x,y)$ in Fig.~\ref{fig:id2}. Hence,
$\mathcal{A}$ is bideterminisable over $\mathbb{Q}(x,y)$, the fraction field of $\mathbb{Q}[X]$. 

\begin{figure}[h!]
\begin{center}
\includegraphics{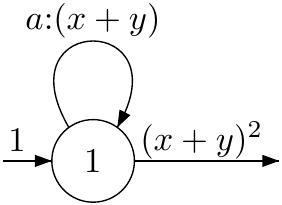}
\end{center}
\caption{\label{fig:id2}The bideterministic weighted automaton $\mathcal{B}$ over $\mathbb{Q}(x,y)$ equivalent to $\mathcal{A}$.}
\end{figure}

In spite of this, we now show that the~automaton $\mathcal{A}$ is not bideterminisable over $\mathbb{Q}[X]$. Assume for~contradiction that
$\|\mathcal{A}\| = \|\mathcal{C}\|$ for some bideterministic weighted automaton $\mathcal{C}$ over $\mathbb{Q}[X]$. 
As $\eps \in \supp(\|\mathcal{A}\|)$, the~automaton $\mathcal{C} = (Q,\sigma,\iota,\tau)$ contains a~state~$q \in Q$ such that both $\iota(q) \neq 0$ and~$\tau(q) \neq 0$,
while $\iota(p) = \tau(p) = 0$ for~all~$p \in Q \setminus \{q\}$. Moreover,
\begin{equation}
\label{eq:id-eq1}
\iota(q) \tau(q) = \left(\|\mathcal{A}\|, \eps\right) = (x + y)^2.
\end{equation} 
Next, as $a \in \supp(\|\mathcal{A}\|)$ as well, there has~to~be a~transition upon~$a$ leading from~$q$ to~a~state with nonzero terminal weight, \emph{i.e.},
an~$a$-labelled self-loop at~$q$. Moreover, by~commutativity of~$\mathbb{Q}[x,y]$,
\begin{displaymath}
\iota(q) \tau(q) \sigma(q,a,q) = \iota(q) \sigma(q,a,q) \tau(q) = \left(\|\mathcal{A}\|, a\right) = (x + y)^3, 
\end{displaymath} 
which together with (\ref{eq:id-eq1}) implies
\begin{equation}
\label{eq:id-eq2}
(x + y)^2 \sigma(q, a, q) = (x + y)^3.
\end{equation} 
However, the only possible $\sigma(q, a, q) \in \mathbb{Q}[x,y]$ satisfying (\ref{eq:id-eq2}) is $\sigma(q, a, q) = x + y$, and~this is not an~element of~$\mathbb{Q}[X]$.
This contradicts our assumption that $\mathcal{C}$ is a~bideterministic weighted automaton over~$\mathbb{Q}[X]$. 

We thus see that the Cardon-Crochemore minimisation algorithm applied to $\mathcal{A}$ outputs a~bideterministic weighted automaton over the~fraction field~$\mathbb{Q}(x,y)$
although the~automaton~$\mathcal{A}$ is not bideterminisable over the~original integral domain~$\mathbb{Q}[X]$.  
\end{example}

\begin{example}
\label{ex:int-domains2}
Consider the weighted automaton $\mathcal{A}$ over the integral domain $\mathbb{Z}$ and alphabet $\Sigma = \{a,b\}$ depicted in Fig. \ref{fig:id3}.
This automaton is clearly bideterministic, so it is trivially bideterminisable.

\begin{figure}[h!]
\begin{center}
\includegraphics{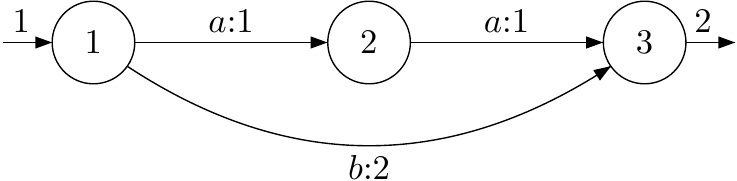}
\end{center}
\caption{\label{fig:id3}The bideterministic weighted automaton $\mathcal{A}$ over $\mathbb{Z}$.}
\end{figure}

Now, the output of the Cardon-Crochemore algorithm depends on how precisely the~left and~right basic languages are computed.
Although the~exact procedure may vary, a~common outcome can be that the~right basic language~$R$ used in~the~reduction step~(\ref{eq:1})
and~the~left basic language~$L$ used in~the~reduction step~(\ref{eq:2}) are both equal to~$\{\eps,a,b\}$. 
In~that case, the~Cardon-Crochemore algorithm produces the~automaton~$\mathcal{B}$ in~Fig.~\ref{fig:id4}.    
  
\begin{figure}[h!]
\begin{center}
\includegraphics{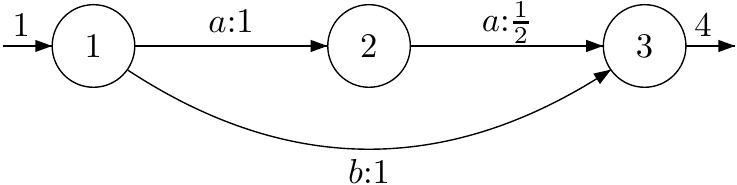}
\end{center}
\caption{\label{fig:id4}The output automaton $\mathcal{B}$ of the Cardon-Crochemore algorithm applied to $\mathcal{A}$.}
\end{figure}  

Although $\mathcal{B}$ is a bideterministic weighted automaton over $\mathbb{Q}$ -- which is the~fraction field of~the~integral domain~$\mathbb{Z}$ -- it contains a weight $1/2$,
so it is not a bideterministic weighted automaton over $\mathbb{Z}$. Of~course, one could construct similar examples for other algorithms of computing the basic languages as well.

We may thus conclude that the Cardon-Crochemore algorithm applied to~a~bideterminisable weighted automaton over an integral domain might not always produce a~bideterministic weighted
automaton over the~same domain -- although it necessarily outputs a~bideterministic weighted automaton over its field of~fractions.
\end{example}

\subsection{Tropical Semirings}

We now establish decidability of the bideterminisability problem for~weighted automata over the~tropical (min-plus) semirings of~nonnegative integers, integers, and~rational numbers -- that is, over the~semi\-rings $\mathbb{N}_{\min} = (\mathbb{N} \cup \{\infty\},\min,+,\infty,0)$,
$\mathbb{Z}_{\min} = (\mathbb{Z} \cup \{\infty\},\min,+,\infty,0)$, and~$\mathbb{Q}_{\min} = (\mathbb{Q} \cup \{\infty\},\min,+,\infty,0)$.  

In~order to obtain these decidability results, we need the following lemma showing that a~very specific form of~determinisability is
decidable for~tropical automata.  

\begin{lemma}
\label{le:trop-deter}
Let $\mathbb{T}$ be one of the semirings $\mathbb{N}_{\min}$, $\mathbb{Z}_{\min}$, or~$\mathbb{Q}_{\min}$.
Let $\mathcal{A}$ be a weighted automaton over~$\mathbb{T}$, and~$\mathcal{B}$ a~deterministic finite automaton without weights.
Then it is decidable whether $\mathcal{A}$ is equivalent to~some deterministic weighted automaton $\mathcal{B}'$ over $\mathbb{T}$ obtained by assigning weights to $\mathcal{B}$.
\end{lemma}
\begin{proof}
Positivity of tropical semirings implies that the automaton $\mathcal{B}'$ can only exist when $\mathcal{B}$ recognises the language $\supp(\|\mathcal{A}\|)$.
This condition is clearly decidable by removing the weights of $\mathcal{A}$ and deciding equivalence of the resulting nondeterministic finite automaton with $\mathcal{B}$.
We may therefore assume that $\mathcal{B}$ indeed recognises $\supp(\|\mathcal{A}\|)$ in what follows.

Moreover, the task is trivial when $\mathcal{B}$ contains no initial state. We therefore assume that $\mathcal{B}$ contains precisely one initial state.

Denote the~unknown weights assigned to $\mathcal{B}$
by $x,y_1,\ldots,y_M,z_1,\ldots,z_N$ for some $M,N \in \mathbb{N}$, where $x$ corresponds to the unknown initial weight,
$y_1,\ldots,y_M$ correspond to the unknown transition weights, and $z_1,\ldots,z_N$ correspond to the unknown terminal weights.
Let $\mathbf{x} = (x,y_1,\ldots,y_M,z_1,\ldots,z_N)$.
Given $w \in \supp(\|\mathcal{A}\|)$, let $\eta_i$ denote, for $i = 1,\ldots,M$, the number of times the unique successful run of $\mathcal{B}$
upon $w$ goes through the transition corresponding to the unknown weight $y_i$. Moreover, for $j = 1,\ldots,N$, let $\nu_j = 1$ if this unique
successful run on $w$ ends in the state corresponding to the unknown weight $z_j$, and let $\nu_j = 0$ otherwise.
Finally, set $\Psi(w) = (1,\eta_1,\ldots,\eta_M,\nu_1,\ldots,\nu_N)$. 

In order for $\mathcal{B}'$ to exist, the unknown weights have to satisfy the equations
$\Psi(w) \cdot \mathbf{x}^T = (\|\mathcal{A}\|, w)$
for all $w \in \supp(\|\mathcal{A}\|)$. If this system has a solution,
then its solution set coincides with the one of a \emph{finite} system of equations
\begin{equation}
\label{eq:system}
\Psi(w_i) \cdot \mathbf{x}^T = \left(\|\mathcal{A}\|, w_i\right) \qquad \text{ for $i = 1,\ldots,K$},
\end{equation}
where $w_1,\ldots,w_K \in \supp(\|\mathcal{A}\|)$ are such that $(\Psi(w_1),\ldots,\Psi(w_K))$ is a~basis 
of~the~vector space over $\mathbb{Q}$ generated by $\Psi(w)$ for $w \in \supp(\|\mathcal{A}\|)$. 
This basis can be effectively obtained, e.g., from the representation of $\{\Psi(w) \mid w \in \supp(\|\mathcal{A}\|)\}$ as~a~semilinear set.
Hence, $w_1,\ldots,w_K$ can be found as well.

We may thus solve the system (\ref{eq:system}) over $\mathbb{N}$, $\mathbb{Z}$, or $\mathbb{Q}$ depending on the semiring considered.
While Gaussian elimination is sufficient to~solve the~system over~$\mathbb{Q}$, the~solution over $\mathbb{Z}$ and~$\mathbb{N}$ requires more sophisticated methods, namely
an~algorithm for~solving systems of~linear Diophantine equations in~the~former case~\cite{schrijver1986a}, and~integer linear programming in~the~latter case~\cite{schrijver1986a}.

If there is no solution, the automaton $\mathcal{B}'$ does not exist. Otherwise, any solution $\mathbf{x}$ gives us a deterministic tropical automaton $\mathcal{B}_{\mathbf{x}}$ obtained
from $\mathcal{B}$ by assigning the weights according to $\mathbf{x}$. By what has been said, either all such automata $\mathcal{B}_{\mathbf{x}}$ are equivalent to $\mathcal{A}$,
or none of them is. Equivalence of~a~deterministic tropical automaton with a nondeterministic one is decidable~\cite{almagor2022a}, so we may take any of the automata $\mathcal{B}_{\mathbf{x}}$  
and decide whether $\|\mathcal{B}_{\mathbf{x}}\| = \|\mathcal{A}\|$. If so, we may set $\mathcal{B}' = \mathcal{B}_{\mathbf{x}}$.
Otherwise, $\mathcal{B}'$ does not exist.
\end{proof}\goodbreak

\begin{remark}
The \emph{determinisability} problem for tropical automata -- whose decidability status is still open~\cite{lombardy2021a,lombardy2006a} -- is clearly recursively enumerable, as one can decide
equivalence of a~nondeterministic tropical automaton with a~deterministic one~\cite{almagor2022a}, as well as enumerate all deterministic tropical automata.
Lemma~\ref{le:trop-deter} implies that a~computable bound on~the~state complexity of~determinisation of~tropical automata would actually be sufficient to establish decidability of this problem:
one could simply enumerate all deterministic finite automata up to the size given by the computable bound, and~check whether weights can be assigned to any of~them such that
one obtains a~tropical automaton equivalent to~the~original one. 
\end{remark} 

We now use Lemma~\ref{le:trop-deter} to establish decidability of the \emph{bideterminisability} problem for weighted automata over tropical semirings. 

\begin{theorem}
\label{th:dec-tropical}
Bideterminisability of weighted automata over the~tropical semirings $\mathbb{N}_{\min}$, $\mathbb{Z}_{\min}$, and~$\mathbb{Q}_{\min}$ is decidable.
\end{theorem}
\begin{proof}
By positivity of tropical semirings, the minimal deterministic finite automaton $\mathcal{B}$ for~$\supp(\|\mathcal{A}\|)$ is bideterministic
whenever a~tropical automaton $\mathcal{A}$ is bideterminisable.
Given $\mathcal{A}$, we may thus remove the~weights and~minimise the automaton to get $\mathcal{B}$. If $\mathcal{B}$ is not bideterministic, $\mathcal{A}$ is not bideterminisable.
If~$\mathcal{B}$ is empty, $\mathcal{A}$ is bideterminisable.
If~$\mathcal{B}$ is bideterministic and nonempty, $\mathcal{A}$ is bideterminisable if and~only if it is equivalent to some $\mathcal{B}'$ obtained from~$\mathcal{B}$ by~assigning
weights to~its transitions, its initial state, and~its terminal state -- and existence of such an~automaton~$\mathcal{B}'$ is decidable by~Lemma~\ref{le:trop-deter}.
\end{proof}

Note that the above described decision algorithm makes use of~deciding equivalence of a~nondeterministic tropical automaton with a deterministic one,
which is $\mathbf{PSPACE}$-complete \cite{almagor2022a}. Nevertheless, this does not rule out existence of~a~more efficient algorithm. We leave the~complexity of~the~bideterminisability problem for tropical automata open.

\subsection{An Undecidability Result}

We now show that the decidability result of Theorem \ref{th:dec-tropical} \emph{does not} generalise to all effective positive semirings.

To this end, let us consider a semiring arising from the semiring of formal languages $2^{\Sigma^*}$ over some alphabet $\Sigma$ by identifying all languages containing at least two different words.
In this way, we obtain a~semi\-ring $(S_{\Sigma}, \sqcup, \cdot, \emptyset, \{\eps\})$ containing all singleton languages over $\Sigma$, the empty language, and~an~element~$\top$ representing the identified languages with two
or more words:
\begin{displaymath}
S_{\Sigma} = \{\emptyset\} \cup {\Sigma^* \choose 1} \cup \{\top\}.
\end{displaymath} 
The operations of $S_{\Sigma}$ are based upon the usual operations on formal languages, \emph{i.e.},
\begin{align*} 
\emptyset \sqcup L = L \sqcup \emptyset = L & \qquad \text{for all } L \in S_{\Sigma}; \\ 
\top \sqcup L = L \sqcup \top = \top & \qquad \text{for all } L \in S_{\Sigma}; \\ 
\{w\} \sqcup \{w\} = \{w\} & \qquad \text{for all } w \in \Sigma^*; \\ 
\{u\} \sqcup \{v\} = \top & \qquad \text{for all distinct } u,v \in \Sigma^*; \\
\emptyset \cdot L = L \cdot \emptyset = \emptyset & \qquad \text{for all } L \in S_{\Sigma}; \\
\top \cdot L = L \cdot \top = \top & \qquad \text{for all } L \in {\Sigma^* \choose 1} \cup \{\top\}; \\
\{u\} \cdot \{v\} = \{uv\} & \qquad \text{for all } u,v \in \Sigma^*.
\end{align*} 
More concisely, $S_{\Sigma}$ can be described as a factor semiring $2^{\Sigma^*} / \equiv$, where $\equiv$ is a congruence on $2^{\Sigma^*}$ such that $K, L \subseteq \Sigma^*$
satisfy $K \equiv L$ if and only if either $K = L$, or $\lvert K\rvert \geq 2$ and $\lvert L\rvert \geq 2$.
It is easy to see that $S_{\Sigma}$ actually is a \emph{positive} semiring.  

We now establish undecidability of the bideterminisability problem for weighted automata over the~semi\-ring~$S_{\Sigma}$ in case $\Sigma$ contains at least two letters. To do so, we reduce
the \emph{Post correspondence problem}~(PCP) to our problem. Recall that a PCP instance over an alphabet $\Sigma$ can be described by a triple $(\Gamma,f,g)$, where~$\Gamma$ is an~alphabet and $f,g\colon \Gamma^* \to \Sigma^*$
are homomorphisms, and the task is to decide whether there exists a nonempty word $w \in \Gamma^+$ such that $f(w) = g(w)$. The~Post correspondence problem is famous for~being undecidable when $\Sigma$ contains at least
two different letters \cite[Section 5.1]{salomaa1985a}.     

\begin{theorem}
Bideterminisability of weighted automata over $S_{\Sigma}$ is undecidable for any alphabet $\Sigma$ containing at least two letters.
\end{theorem}     
\begin{proof}
Consider an arbitrary PCP instance $(\Gamma,f,g)$ over the alphabet $\Sigma$. Let $\Gamma = \{c_1,\ldots,c_m\}$, and let us construct a weighted automaton
$\mathcal{A}_{\Gamma,f,g}$ over the semiring $S_{\Sigma}$ and alphabet $\Gamma$ depicted in~Fig.~\ref{fig:undec}.
\begin{figure}[h!]
\begin{center}
\includegraphics{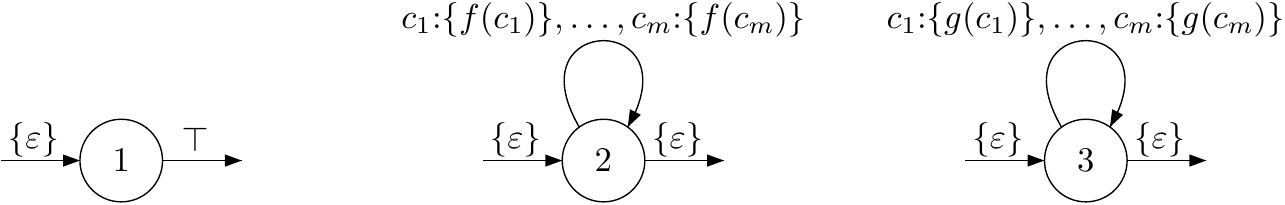}
\end{center}
\caption{\label{fig:undec}The weighted automaton $\mathcal{A}_{\Gamma,f,g}$ over $S_{\Sigma}$ and $\Gamma$ corresponding to the PCP instance $(\Gamma,f,g)$.}
\end{figure}

Given $w \in \Gamma^*$, clearly
\begin{equation}
\label{eq:pcp-behaviour}
\left(\left\|\mathcal{A}_{\Gamma,f,g}\right\|, w\right) = \left\{\begin{array}{ll} \top & \text{if $w = \eps$}, \\ \{x\} & \text{if $w \in \Gamma^+$ and $f(w) = g(w) = x$}, \\ \top & \text{if $w \in \Gamma^+$ and $f(w) \neq g(w)$}.\end{array}\right.
\end{equation} 
We prove undecidability of our problem by showing that the automaton $\mathcal{A}_{\Gamma,f,g}$ is bideterminisable if and~only if there is no solution to the PCP instance $(\Gamma,f,g)$, \emph{i.e.}, there is no
$w \in \Gamma^+$ such that $f(w) = g(w)$. 

\begin{figure}[h!]
\begin{center}
\includegraphics{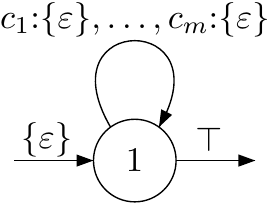}
\end{center}
\caption{\label{fig:undec2}The bideterministic weighted automaton $\mathcal{B}$ equivalent to $\mathcal{A}_{\Gamma,f,g}$ in case $(\Gamma,f,g)$ has no solution.}
\end{figure}
Indeed, if there is no solution to $(\Gamma,f,g)$, then (\ref{eq:pcp-behaviour}) implies that $(\|\mathcal{A}_{\Gamma,f,g}\|, w) = \top$ for all $w \in \Gamma^*$. The~automaton $\mathcal{A}_{\Gamma,f,g}$
is equivalent to the bideterministic weighted automaton $\mathcal{B}$ in Fig.~\ref{fig:undec2} as a result. 

Now, suppose that there is some solution $w \in \Gamma^+$ to the instance $(\Gamma,f,g)$, and assume for contradiction that $\mathcal{A}_{\Gamma,f,g}$ is equivalent to some
bideterministic weighted automaton $\mathcal{C} = (Q,\sigma,\iota,\tau)$ over $S_{\Sigma}$ and $\Gamma$. As~$\supp(\|\mathcal{A}_{\Gamma,f,g}\|) \neq \emptyset$, the automaton $\mathcal{C}$
contains precisely one state with initial weight different from~$\emptyset$ -- denote this state by $q$. 
As it is implied by (\ref{eq:pcp-behaviour}) that \smash{$(\|\mathcal{A}_{\Gamma,f,g}\|, w) \in {\Sigma^* \choose 1}$}, necessarily
\begin{displaymath}
\iota(q) \in {\Sigma^* \choose 1},
\end{displaymath}
and as $(\|\mathcal{A}_{\Gamma,f,g}\|, \eps) = \top$, this also implies that
\begin{displaymath}
\tau(q) = \top.
\end{displaymath}
However, the bideterministic weighted automaton $\mathcal{C}$ cannot have any other state with terminal weight different from $\emptyset$,
which means that $(\|\mathcal{A}_{\Gamma,f,g}\|, w) \in \{\emptyset, \top\}$, contradicting the fact that $(\|\mathcal{A}_{\Gamma,f,g}\|, w)$ is a~singleton language.  
\end{proof}

\section{Conclusions}

The concept of \emph{bideterminism} has been generalised to weighted automata over a~semiring. We have seen that
trim bideterministic weighted automata over integral domains and over positive semirings are always minimal, generalising the~well-known property
of bideterministic finite automata without weights~\cite{tamm2003a,tamm2004a}. On the contrary, we have observed that this property
does not hold over other than zero-divisor free semirings, and~thus also over nontrivial commutative rings other than integral domains.
For a~large class of commutative semirings $S$ including also many finite commutative rings, we have shown that a~bideterministic weighted automaton
over~$S$ might not even admit an~equivalent minimal automaton that is bideterministic as well, and~we have completely characterised all $m \in \mathbb{N}$ such
that this is the~case over the~ring~$\mathbb{Z}/m\mathbb{Z}$.

We also studied the \emph{bideterminisability} problem for weighted automata, in which the task is to decide whether a given weighted automaton over a~semiring~$S$
admits an~equivalent bideterministic automaton over~$S$. We have seen that this problem is decidable over fields by simply running the Cardon-Crochemore minimisation
algorithm and checking whether its output is bideterministic or not. On the other hand, we have noted that this simple decision procedure is no longer sufficient for
weighted automata over integral domains. Moreover, we have established decidability of the bideterminisability problem for tropical automata -- and observed that this result
does not generalise to automata over effective positive semirings.

We have left several questions open. For instance: Can one characterise commutative rings over which bideterministic weighted automata always admit equivalent minimal automata that are bideterministic as~well?
Is the bideterminisability problem decidable for automata over integral domains -- or at least over, e.g., completely integrally closed domains?
What can one say about the~computational complexity of~the~bideterminisability problem over tropical semirings? And can one characterise positive semirings with decidable bideterminisability problem?


\bibliographystyle{abbrv}
\bibliography{references}

\begin{thebibliography}{10}

\bibitem{allauzen2003a}
C.~Allauzen and M.~Mohri.
\newblock Efficient algorithms for testing the twins property.
\newblock {\em Journal of Automata, Languages and Combinatorics},
  8(2):117--144, 2003.

\bibitem{almagor2022a}
S.~Almagor, U.~Boker, and O.~Kupferman.
\newblock What's decidable about weighted automata?
\newblock {\em Information and Computation}, 282, 2022.
\newblock Article 104651.

\bibitem{anderson2011a}
D.~F. Anderson, M.~C. Axtell, and J.~A. Stickles.
\newblock Zero-divisor graphs in commutative rings.
\newblock In M.~Fontana, S.-E. Kabbaj, B.~Olberding, and I.~Swanson, editors,
  {\em Commutative Algebra: {N}oetherian and Non-{N}oetherian Perspectives},
  pages 23--45. Springer, 2011.

\bibitem{anderson1999a}
D.~F. Anderson and P.~S. Livingston.
\newblock The zero-divisor graph of a commutative ring.
\newblock {\em Journal of Algebra}, 217(2):434--447, 1999.

\bibitem{angluin1982a}
D.~Angluin.
\newblock Inference of reversible languages.
\newblock {\em Journal of the {ACM}}, 29(3):741--765, 1982.

\bibitem{beal2005a}
M.-P. B{\'e}al, S.~Lombardy, and J.~Sakarovitch.
\newblock On the equivalence of {$\mathbb{Z}$}-automata.
\newblock In {\em Automata, Languages and Programming, ICALP 2005}, pages
  397--409. Springer, 2005.

\bibitem{beal2006a}
M.-P. B{\'e}al, S.~Lombardy, and J.~Sakarovitch.
\newblock Conjugacy and equivalence of weighted automata and functional
  transducers.
\newblock In {\em Computer Science -- Theory and Applications, CSR 2006}, pages
  58--69. Springer, 2006.

\bibitem{bell2021a}
J.~Bell and D.~Smertnig.
\newblock Noncommutative rational {P\'o}lya series.
\newblock {\em Selecta Mathematica}, 27(3):article 34, 2021.

\bibitem{bell2023a}
J.~P. Bell and D.~Smertnig.
\newblock Computing the linear hull: Deciding {D}eterministic? and
  {U}nambiguous? for weighted automata over fields.
\newblock In {\em Logic in Computer Science, LICS 2023}, 2023.

\bibitem{berstel2011a}
J.~Berstel and C.~Reutenauer.
\newblock {\em Noncommutative Rational Series with Applications}.
\newblock Cambridge University Press, 2011.

\bibitem{cardon1980a}
A.~Cardon and M.~Crochemore.
\newblock D{\'e}termination de la repr{\'e}sentation standard d'une s{\'e}rie
  reconnaissable.
\newblock {\em Informatique Th{\'e}orique et Applications}, 14(4):371--379,
  1980.

\bibitem{ciric2010a}
M.~{\'C}iri{\'c}, M.~Droste, J.~Ignjatovi{\'c}, and H.~Vogler.
\newblock Determinization of weighted finite automata over strong bimonoids.
\newblock {\em Information Sciences}, 180:3497--3520, 2010.

\bibitem{demeyer2005a}
F.~DeMeyer and L.~DeMeyer.
\newblock Zero divisor graphs of semigroups.
\newblock {\em Journal of Algebra}, 283(1):190--198, 2005.

\bibitem{dolzan2012a}
D.~Dol{\v z}an and P.~Oblak.
\newblock The zero-divisor graphs of rings and semirings.
\newblock {\em International Journal of Algebra and Computation}, 22(4), 2012.
\newblock Article 1250033.

\bibitem{droste2009a}
M.~Droste, W.~Kuich, and H.~Vogler, editors.
\newblock {\em Handbook of Weighted Automata}.
\newblock Springer, 2009.

\bibitem{droste2021a}
M.~Droste and D.~Kuske.
\newblock Weighted automata.
\newblock In J.-{\'E}. Pin, editor, {\em Handbook of Automata Theory, Vol. 1},
  chapter~4, pages 113--150. European Mathematical Society, 2021.

\bibitem{eilenberg1974a}
S.~Eilenberg.
\newblock {\em Automata, Languages, and Machines, Vol. A}.
\newblock Academic Press, 1974.

\bibitem{golan1999a}
J.~S. Golan.
\newblock {\em Semirings and their Applications}.
\newblock Kluwer Academic Publishers, 1999.

\bibitem{gruber2012a}
H.~Gruber.
\newblock Digraph complexity measures and applications in formal language
  theory.
\newblock {\em Discrete Mathematics \& Theoretical Computer Science},
  14(2):189--204, 2012.

\bibitem{hebisch1998a}
U.~Hebisch and H.~J. Weinert.
\newblock {\em Semirings}.
\newblock World Scientific, 1998.

\bibitem{janin2015a}
D.~Janin.
\newblock Free inverse monoids up to rewriting.
\newblock Technical report, La{BRI} -- Laboratoire {B}ordelais de Recherche en
  Informatique, 2015.
\newblock Available at \texttt{https://hal.archives-ouvertes.fr/hal-01182934}.

\bibitem{kirsten2014a}
D.~Kirsten.
\newblock An algebraic characterization of semirings for which the support of
  every recognizable series is recognizable.
\newblock {\em Theoretical Computer Science}, 534:45--52, 2014.

\bibitem{kirsten2009a}
D.~Kirsten and S.~Lombardy.
\newblock Deciding unambiguity and sequentiality of polynomially ambiguous
  min-plus automata.
\newblock In {\em Symposium on Theoretical Aspects of Computer Science, STACS
  2009}, pages 589--600, 2009.

\bibitem{kirsten2005a}
D.~Kirsten and I.~M{\"a}urer.
\newblock On the determinization of weighted automata.
\newblock {\em Journal of Automata, Languages and Combinatorics},
  10(2--3):287--312, 2005.

\bibitem{kostolanyi2022b}
P.~Kostol{\'a}nyi.
\newblock Bideterministic weighted automata.
\newblock In {\em Algebraic Informatics, CAI 2022}, pages 161--174, 2022.

\bibitem{kostolanyi2022a}
P.~Kostol{\'a}nyi.
\newblock Determinisability of unary weighted automata over the rational
  numbers.
\newblock {\em Theoretical Computer Science}, 898:110--131, 2022.

\bibitem{lombardy2021a}
S.~Lombardy and J.~Mairesse.
\newblock Max-plus automata.
\newblock In J.-{\'E}. Pin, editor, {\em Handbook of Automata Theory, Vol. 1},
  chapter~5, pages 151--188. European Mathematical Society, 2021.

\bibitem{lombardy2006a}
S.~Lombardy and J.~Sakarovitch.
\newblock Sequential?
\newblock {\em Theoretical Computer Science}, 356:224--244, 2006.

\bibitem{mcnaughton1967a}
R.~McNaughton.
\newblock The loop complexity of pure-group events.
\newblock {\em Information and Control}, 11(1--2):167--176, 1967.

\bibitem{mcnaughton1969a}
R.~McNaughton.
\newblock The loop complexity of regular events.
\newblock {\em Information Sciences}, 1(3):305--328, 1969.

\bibitem{mohri1997a}
M.~Mohri.
\newblock Finite-state transducers in language and speech processing.
\newblock {\em Computational Linguistics}, 23(2):269--311, 1997.

\bibitem{mohri2009a}
M.~Mohri.
\newblock Weighted automata algorithms.
\newblock In M.~Droste, W.~Kuich, and H.~Vogler, editors, {\em Handbook of
  Weighted Automata}, chapter~6, pages 213--256. Springer, 2009.

\bibitem{myers2021a}
R.~S.~R. Myers, S.~Milius, and H.~Urbat.
\newblock Nondeterministic syntactic complexity.
\newblock In {\em Foundations of Software Science and Computation Structures,
  FOSSACS 2021}, pages 448--468, 2021.

\bibitem{pin1992a}
J.-{\'E}. Pin.
\newblock On reversible automata.
\newblock In {\em Latin American Symposium on Theoretical Informatics, LATIN
  1992}, pages 401--416, 1992.

\bibitem{polak2005a}
L.~Pol{\'a}k.
\newblock Minimalizations of {NFA} using the universal automaton.
\newblock {\em International Journal of Foundations of Computer Science},
  16(5):999--1010, 2005.

\bibitem{sakarovitch2009a}
J.~Sakarovitch.
\newblock {\em Elements of Automata Theory}.
\newblock Cambridge University Press, 2009.

\bibitem{sakarovitch2009b}
J.~Sakarovitch.
\newblock Rational and recognisable power series.
\newblock In M.~Droste, W.~Kuich, and H.~Vogler, editors, {\em Handbook of
  Weighted Automata}, chapter~4, pages 105--174. Springer, 2009.

\bibitem{salomaa1985a}
A.~Salomaa.
\newblock {\em Computation and Automata}.
\newblock Cambridge University Press, 1985.

\bibitem{schrijver1986a}
A.~Schrijver.
\newblock {\em Theory of Linear and Integer Programming}.
\newblock John Wiley \& Sons, 1986.

\bibitem{schutzenberger1961a}
M.-P. Sch{\"u}tzenberger.
\newblock On the definition of a family of automata.
\newblock {\em Information and Control}, 4(2--3):245--270, 1961.

\bibitem{shankar2003a}
P.~Shankar, A.~Dasgupta, K.~Deshmukh, and B.~S. Rajan.
\newblock On viewing block codes as finite automata.
\newblock {\em Theoretical Computer Science}, 290(3):1775--1797, 2003.

\bibitem{stephen1990a}
J.~B. Stephen.
\newblock Presentations of inverse monoids.
\newblock {\em Journal of Pure and Applied Algebra}, 63(1):81--112, 1990.

\bibitem{tamm2008a}
H.~Tamm.
\newblock On transition minimality of bideterministic automata.
\newblock {\em International Journal of Foundations of Computer Science},
  19(3):677--690, 2008.

\bibitem{tamm2003a}
H.~Tamm and E.~Ukkonen.
\newblock Bideterministic automata and minimal representations of regular
  languages.
\newblock In {\em Implementation and Application of Automata, CIAA 2003}, pages
  61--71, 2003.

\bibitem{tamm2004a}
H.~Tamm and E.~Ukkonen.
\newblock Bideterministic automata and minimal representations of regular
  languages.
\newblock {\em Theoretical Computer Science}, 328(1--2):135--149, 2004.

\end{thebibliography}

\end{document}